\PassOptionsToPackage{table}{xcolor}
\documentclass[%
 reprint,
superscriptaddress,
nofootinbib,
 amsmath,amssymb,
 aps,
pra,
]{revtex4-2}

\usepackage{tikz} 
\usepackage{amsmath,amsfonts,amsthm, mathtools}
\usepackage{float}
\usepackage{dsfont}
\usepackage{csquotes}
\usepackage{amssymb}
\usepackage{braket}
\usepackage[table]{xcolor}
\usepackage{verbatim}
\usepackage{mathtools}
\usepackage{rotating}
\usepackage{thmtools, thm-restate}
\usepackage{bm}
\usepackage{hyperref}
\usepackage{mathtools}
\usepackage{mathbbol}
\usepackage{makecell}
\usepackage{pifont}
\usepackage{appendix}
\usepackage{diagbox}      
\usepackage{soul}

\definecolor{darkred}  {rgb}{0.5,0,0}
\definecolor{darkblue} {rgb}{0,0,0.5}
\definecolor{darkgreen}{rgb}{0,0.5,0}
\hypersetup{
  colorlinks = true,
  urlcolor  = blue,         
  linkcolor = darkblue,     
  citecolor = darkgreen,    
  filecolor = darkred    
}
\theoremstyle{definition}

\newtheorem{corollary}{Corollary}

\newtheorem{lemma}{Lemma}
\newtheorem{proposition}{Proposition}

\newtheorem{theorem}{Theorem}

\newtheorem{example}{Example}

\newtheorem*{remark}{Remark}

\newcommand{\conv}{\text{conv}}

\newcommand{\mbb}{\mathbb}
\newcommand{\mc}{\mathcal}
\newcommand{\msf}{\mathsf}

\newcommand{\tr}{\textrm{Tr}}

\usepackage{multirow}
\newcommand{\op}[2]{|#1\rangle\langle#2|}

\definecolor{cool_green}{rgb}{0.0, 0.5, 0.0}
\newcommand{\norm}[1]{\left\lVert#1\right\rVert}
\usepackage{blkarray}
\usepackage{adjustbox}
\usepackage{graphicx}
\usepackage{dcolumn}
\usepackage{bm}

\usepackage[normalem]{ulem}

\begin{document}

\preprint{APS/123-QED}

\title{Exact Incompatibility-Breaking Criterion for Unital Qubit Channels}

\author{Yujie Zhang}
\email{yujie4physics@gmail.com}
\affiliation{Institute for Quantum Computing and Department of Physics \& Astronomy,
University of Waterloo, 200 University Ave W, Waterloo, Ontario, N2L 3G1, Canada}
\affiliation{Perimeter Institute for Theoretical Physics, 31 Caroline Street North, Waterloo, Ontario, Canada N2L 2Y5}

\date{\today}

\begin{abstract}
We study when a noisy qubit channel renders all positive-operator-valued measures (POVMs) jointly measurable. For every unital qubit channel, we derive the exact incompatibility-breaking criterion and prove that projective measurements already determine the boundary for arbitrary POVMs. Through the steering--joint-measurability correspondence, this result gives the exact POVM-steering boundary for all two-qubit states with maximally mixed marginals, with the Werner states recovered as a special case. For nonunital qubit channels, we construct a generally asymmetric parent POVM and obtain an explicit sufficient incompatibility-breaking condition, which in turn yields a sufficient unsteerability criterion for arbitrary two-qubit states.

\end{abstract}

\maketitle


Quantum measurements are not all jointly measurable. This simple
fact is one of the operational roots of nonclassicality
\cite{Heinosaari2016,Guhne2023}. In Bell scenarios, measurement incompatibility is necessary for nonlocality, although the converse fails in general~\cite{Quintino2015,CavalcantiSkrzypczyk2016, Hirsch2018,BeneVertesi2018} and is recovered only in particular restricted or multipartite settings
\cite{Wolf2009,Plavala2025}. In steering scenarios~\cite{Uola2020}, joint-measurability problems can be mapped one-to-one to steering problems~\cite{Quintino2014,Uola2014,Uola2015,Cavalcanti2016, Kiukas2017}. Measurement incompatibility is also closely related to generalized contextuality \cite{Spekkens2005,TavakoliUola2020,Selby2023, ZhangSchmidYingSpekkens2026,ZhangYingSchmid2025}, and admits operational characterizations through state-discrimination
and programmable-measurement tasks
\cite{Carmeli2019a,Skrzypczyk2019,Buscemi2020}.

A family of POVMs $\{\msf M^{(x)}\}_x$, with
$\msf M^{(x)}=\{M_{a|x}\}_a$, is jointly measurable if there exists a single parent POVM $\Pi=\{\Pi_\lambda\}_\lambda$, together with
classical postprocessings $p(a|x,\lambda)$, such that
\begin{align}
M_{a|x}=\sum_\lambda p(a|x,\lambda)\Pi_\lambda
\qquad
\forall a,x.
\end{align}
Otherwise, $\{\msf M^{(x)}\}_x$ is incompatible.


An important version of the compatibility problem asks whether a
quantum channel renders an entire measurement class compatible. We model the noise by a completely positive, trace-preserving (CPTP) map $\mc N$, acting on measurement effects
through its Heisenberg adjoint $\mc N^\dagger$. For a measurement class $\mc C$, the channel is $\mc C$-incompatibility breaking if $\{\mc N^\dagger(\msf M^{(x)})\}_x$ is jointly measurable for every finite family $\{\msf M^{(x)}\}_x\subseteq\mc C$ \cite{Heinosaari2015, zhang2026parallel}. Taking $\mc C$ to be all PVMs or all POVMs defines PVM- and POVM-incompatibility-breaking channels, respectively. PVMs describe sharp measurements, whereas POVMs are the most general measurements allowed by quantum theory and can offer
advantages over PVMs in several information-theoretic tasks
\cite{Peres1988,OszmaniecBiswas2019,MassarPopescu1995,VertesiBene2010}. It is therefore natural to ask whether nonprojective measurements are also more robustly incompatible under noise. Since PVMs form a subclass of POVMs, every POVM-incompatibility-breaking channel is PVM-incompatibility-breaking. The central question is therefore whether PVM-incompatibility breaking already implies POVM-incompatibility breaking.

The canonical example is the qubit depolarizing channel $\Delta_r(X):=rX+(1-r){\tr[X]}\mbb{1}/2$. Determining its PVM- and
POVM-incompatibility-breaking thresholds is equivalent, through the
steering--joint-measurability correspondence
\cite{Quintino2014,Uola2014}, to the Werner problem
\cite{Werner2014}. The latter asks for which
values of $r$ the Werner state
\begin{align}
\rho_W(r) =r\op{\Psi^-}{\Psi^-}+(1-r)\frac{\mbb 1\otimes \mbb 1}{4}
\end{align}
admits a local-hidden-state model for all PVMs or all POVMs. Werner
constructed such a model for PVMs up to $r=1/2$~\cite{Werner1989}, whereas
Barrett obtained a model for arbitrary POVMs only up to
$r=5/12$~\cite{Barrett2002}. This gap was recently closed
independently in Refs.~\cite{Zhang2024,Renner2024}, which established
the exact POVM threshold $r=1/2$ by two different constructive proofs. Consequently, the depolarizing channel is POVM-incompatibility
breaking if and only if it is PVM-incompatibility breaking. 


Depolarizing noise is only the isotropic member of the qubit channels, just as Werner states form only a one-parameter family of two-qubit states. The natural question is therefore whether the PVM--POVM equivalence survives beyond this highly symmetric setting. On the channel side, this amounts to studying incompatibility breaking under general qubit noise; through the steering--joint-measurability correspondence, the same problem becomes that of determining POVM-steering criteria for general two-qubit states.


Inspired by Refs.~\cite{Nguyen2018, Nguyen2019, ZhangZhangChitambar2025}, we develop a general dual framework for studying the simulability of a set of noisy POVMs by a prescribed parent POVM. Rather than optimizing directly over classical postprocessings as in the usual primal, constructive approach to measurement compatibility, we characterize the convex set of POVMs simulable from the parent POVM through its support function. The universal simulability problem is then equivalent to ruling out every separating witness. For qubit channels, the resulting witness problem further reduces to a geometric inequality on the Bloch sphere.

For a unital qubit channel with Bloch matrix $D$, this reduction yields the exact incompatibility-breaking criterion $2\int_{S^2}\norm{D\hat n}d\mu(\hat n)\le1$, where $d\mu(\hat n)$ is the normalized rotation-invariant measure on $S^2$. Since this is precisely the known PVM criterion, PVM- and POVM-incompatibility breaking coincide for every unital qubit channel. Through the steering--joint-measurability correspondence, the same result gives the exact POVM-steering boundary for all two-qubit states with maximally mixed marginals, closing the gap between the known PVM criterion and previous sufficient conditions for POVMs~\cite{Jevtic2015,Nguyen2016,Nguyen2020}.

For nonunital qubit channels, we construct a generally asymmetric parent POVM and derive a sufficient incompatibility-breaking condition. Equivalently, this gives a sufficient criterion for arbitrary two-qubit states to be unsteerable from Alice to Bob under all POVMs.

In higher dimensions, the dual formulation reduces the question of whether the Haar parent at the PVM threshold also simulates all POVMs to an explicit inequality involving positive operators. Proving this inequality would extend the exact PVM--POVM equivalence from qubits to depolarizing channels in arbitrary dimension.


\textit{Support-function formulation--} We first compare the two convex sets entering the simulation problem: the
set of noisy POVMs and the set of POVMs obtainable by postprocessing a
prescribed parent. Requiring every noisy POVM to be simulable by a fixed parent POVM  is equivalent to the inclusion of the two convex sets. 

Let $(\Lambda,\nu)$ be a measure space, and 
$\Pi=\{\Pi_\lambda d\nu(\lambda)\}_{\lambda\in\Lambda}$ be a parent
POVM with $\Pi_\lambda\ge\mbb{0}$ and
$\int_\Lambda \Pi_\lambda d\nu(\lambda)=\mbb{1}$. A finite-outcome POVM
$\msf M=\{M_a\}_a$ is simulable from $\Pi$ if 
\begin{align}
M_a=\int_\Lambda p(a|\lambda)\Pi_\lambda d\nu(\lambda)
\label{eq:simulation}
\end{align}
for some classical postprocessing satisfying $p(a|\lambda)\ge0$ and $\sum_a p(a|\lambda)=1$.

For a fixed finite outcome set, let $\mathcal S_{\Pi}$ denote the compact convex set of finite-outcome POVMs simulable from $\Pi$. For every finite Hermitian family $Y=\{Y_a\}_a$, the support-function characterization of $\mathcal S_{\Pi}$ gives
\begin{align}
\msf M\in\mathcal S_{\Pi}\Longleftrightarrow\sum_a\tr(Y_aM_a)
\le h_{\mathcal S_{\Pi}}(Y)~~\forall Y=\{Y_a\}_a
\end{align}
where
$h_{\mathcal S_{\Pi}}(Y):=\sup_{\msf M'\in\mathcal S_{\Pi}}\sum_a\tr(Y_aM'_a)$.
Substituting Eq.~\eqref{eq:simulation} and optimizing
$p(a|\lambda)$ pointwise gives
\begin{align}
h_{\mathcal S_{\Pi}}(Y)&=\sup_{p(a|\lambda)}
\int_\Lambda\sum_a p(a|\lambda)\tr(Y_a\Pi_\lambda)d\nu(\lambda)
\notag\\
&=\int_\Lambda\max_a\tr(Y_a\Pi_\lambda)d\nu(\lambda).
\label{eq:support-general}
\end{align}

We now apply this characterization to the noisy POVMs $\{\mc N^\dagger(M_a)\}_a$, where $\mc N$ is a CPTP map and $\{M_a\}_a$ ranges over all finite-outcome POVMs. For a fixed witness family,
their support function is
\begin{align}
\sup_{\{M_a\}}\sum_a\tr\left[Y_a\mc N^\dagger(M_a)\right]&=\sup_{\{M_a\}}\sum_a\tr\left[\mc N(Y_a)M_a\right]\label{eq:pri-dual}\\
&=\inf_Z\left\{\tr Z:Z\ge\mc N(Y_a)\ \forall a\right\}.\notag
\end{align}
The last equality follows from SDP duality, which holds
because the primal problem is strictly feasible. Hence every $\{\mc N^\dagger(M_a)\}_a$ is simulable from $\Pi$ if and
only if
\begin{align}
\inf_Z\left\{\tr Z:Z\ge\mc N(Y_a)\ \forall a\right\}\le\int_\Lambda\max_a\tr(Y_a\Pi_\lambda)d\nu(\lambda)
\label{eq:dual-charac}
\end{align}
for every finite Hermitian family $\{Y_a\}_a$.

\begin{remark}
Every continuous-outcome POVM on a $d$-dimensional system is a classical randomization of POVMs with at most $d^2$ outcomes \cite{Chiribella2007,Ariano2005}. Since this classical randomization can be absorbed into the postprocessing of the parent POVM, it suffices to test
Eq.~\eqref{eq:dual-charac} for families $\{Y_a\}_{a=1}^{m}$ with $m\leq d^2$.
\end{remark}

For invertible $\mc N$, complementary slackness of the SDPs in Eq.~\eqref{eq:pri-dual} gives the following reduction.
\begin{proposition}
\label{prop:orthogonality-normal-form}
Let $\mc N$ be an invertible CPTP map and $\Pi=\{\Pi_\lambda d\nu(\lambda)\}_{\lambda\in\Lambda}$ be a parent
POVM. Then $\Pi$ simulates
$\{\mc N^\dagger(M_a)\}_a$ for every POVM $\{M_a\}_a$ if and only if
\begin{align}
\int_\Lambda\max_a\tr \left[-\mc N^{-1}(P_a)\Pi_\lambda\right]
d\nu(\lambda)\ge0
\label{eq:orthogonality-normal-form}
\end{align}
for every finite family $P_a\ge\mbb{0}$ for which there exists a POVM
$\{Q_a\}_a$ satisfying
\begin{align}
P_aQ_a=\mbb{0}
\qquad
\forall a.
\label{eq:orthogonality-condition}
\end{align}
\end{proposition}

The proof is given in Appendix~\ref{app:orthogonality-normal-form}. Briefly, let $Z_*$ and $\{Q_a\}_a$ be optimal solutions of the dual and primal SDPs in Eq.~\eqref{eq:pri-dual}, respectively. For a witness family $\{Y_a\}_a$, define $P_a:=Z_*-\mc N(Y_a)$. Dual feasibility gives $P_a\geq\mbb 0$, while complementary slackness gives
$P_aQ_a=\mbb 0$. Conversely, any family $\{P_a\}_a$ satisfying these positivity and orthogonality conditions defines a witness through $Y_a=-\mc N^{-1}(P_a)$. If
such a family violates Eq.~\eqref{eq:orthogonality-normal-form}, then $\{Y_a\}_a$ witnesses the failure of universal simulation by the prescribed parent POVM.

Proposition~\ref{prop:orthogonality-normal-form} certifies a prescribed parent POVM but does not select one. The applications below therefore involve two steps: selecting a suitable parent and then proving, through the dual criterion, that it simulates the entire family of noisy POVMs.


\textit{Unital qubit channels--} We now apply Proposition~\ref{prop:orthogonality-normal-form} to
unital qubit channels. Let $\mc N_D$ be a unital qubit channel whose
Heisenberg adjoint has the Bloch representation
\begin{align}
\mc N_D^\dagger\left(x_0\mbb{1}+\vec x\cdot\vec\sigma\right)=x_0\mbb{1}+(D\vec x)\cdot\vec\sigma,
\label{eq:unital-heisenberg-channel}
\end{align}
where $D$ is a real $3\times3$ matrix.
\begin{theorem}
\label{thm:main}
A unital qubit channel $\mc N_D$ is incompatibility breaking for
arbitrary POVMs if and only if 
\begin{align}
2\int_{S^2}\norm{D\hat n}d\mu(\hat n)\le1.
\label{eq:unital-main-condition}
\end{align}
\end{theorem}
\begin{proof}
Necessity follows from the known PVM-incompatibility-breaking
condition~\cite{Nguyen2016}, since PVMs form a subclass of POVMs.  For sufficiency, define
\begin{align}
m_D:=\int_{S^2}\norm{D\hat n}d\mu(\hat n), \qquad \delta_D:=\frac{1}{2}-m_D\ge0,
\label{eq:unital-norm}
\end{align}
and consider the parent POVM
$\{\Pi_D(\hat n)d\mu(\hat n)\}_{\hat n\in S^2}$, where
$d\mu(\hat n)$ is the rotation-invariant measure and
\begin{align}
\Pi_D(\hat n):=2\left[\left(\norm{D\hat n}+\delta_D\right)\mbb1+(D\hat n)\cdot\vec\sigma\right].
\label{eq:ellipsoidal-parent}
\end{align} 
This defines a parent POVM: positivity follows from $\norm{D\hat n}+\delta_D\ge\norm{D\hat n}$, and normalization follows
from Eq.~\eqref{eq:unital-norm}. Moreover,
\begin{align}
\int_{\hat u\cdot\hat n\ge0}
\Pi_D(\hat n)d\mu(\hat n)=\frac{1}{2}\left[\mbb1+(D\hat u)\cdot\vec\sigma\right],
\label{eq:ellipsoid-hemisphere}
\end{align}
so $\Pi_D(\hat n)$ simulates $\{\mc N_D^{\dagger}(M_a)\}_a$ for every qubit PVM $\{M_a\}_a$.

We first assume that $D$ is invertible. The same explicit parent POVM for singular $D$ follows from the limiting argument in Appendix~\ref{app:singular-general}.

By Proposition~\ref{prop:orthogonality-normal-form}, it is now sufficient to prove
\begin{align}
\int_{S^2}\max_a\tr\left[-\mc N_D^{-1}(P_a)\Pi_D(\hat n)\right]d\mu(\hat n)\ge 0
\label{eq:unital-positive-form}
\end{align}
for every finite family $P_a\ge\mbb{0}$ for which there exists a POVM
$\{Q_a\}_a$ satisfying $P_aQ_a=\mbb{0}$ for all $a$.

If $P_a=\mbb{0}$ for some $a$, the maximum in
Eq.~\eqref{eq:unital-positive-form} contains the zero branch, and the
claim is immediate. Now let
\begin{align}
\mc A:=\{a:Q_a\ne\mbb{0}\}.
\end{align}
Restricting the maximum in Eq.\eqref{eq:unital-positive-form} to $\mc A$ can only decrease it. For every $a\in\mc A$, positivity of $P_a$, $Q_a$ and $P_aQ_a=\mbb{0}$ imply
\begin{align}
P_a=\alpha_a\left(\mbb{1}-\hat u_a\cdot\vec\sigma\right),\qquad Q_a=
\beta_a\left(\mbb{1}+\hat u_a\cdot\vec\sigma\right),
\label{eq:qubit-PQ-form}
\end{align}
for some $\alpha_a,\beta_a>0$ and $\hat u_a\in S^2$. The normalization
$\sum_{a\in\mc A}Q_a=\mbb{1}$ gives $\sum_{a\in\mc A}\beta_a=1$, $\sum_{a\in\mc A}\beta_a\hat u_a=\vec0$, and
\begin{align}
\vec0\in\conv\{\hat u_a:a\in\mc A\}.
\label{eq:unital-balance}
\end{align}

With the convention in Eq.~\eqref{eq:unital-heisenberg-channel}, one has 
\begin{align}
\mc N_D^{-1}(P_a)=\alpha_a\left[\mbb1-(D^{-T}\hat u_a)\cdot\vec\sigma\right].
\end{align}
Combining this with Eq.~\eqref{eq:ellipsoidal-parent} gives
\begin{align}
\tr\left[-\mc N_D^{-1}(P_a)\Pi_D(\hat n)\right]=4\alpha_a\left[\hat u_a\cdot\hat n-\norm{D\hat n}-\delta_D\right].
\label{eq:unital-geometric-reduction}
\end{align}
Define the nonnegative and even function
\begin{subequations}
\begin{align}
&h_D(\hat n):=\norm{D\hat n}+\delta_D\\
&\int_{S^2}h_D(\hat n)d\mu(\hat n)=\frac{1}{2}.
\label{eq:average condition}
\end{align}
\end{subequations}
Thus, the left-hand side of Eq.~\eqref{eq:unital-positive-form} is bounded below by
\begin{align}
4\int_{S^2}\max_{a\in\mc A}\alpha_a\left[\hat u_a\cdot\hat n-h_D(\hat n)\right]d\mu(\hat n)\ge0.
\end{align}
The last inequality follows from Lemma~\ref{lem:weighted-ellipsoidal-average}, using only the established conditions in Eqs.~\eqref{eq:average condition} and \eqref{eq:unital-balance}.

Thus, the PVM- and POVM-incompatibility-breaking criteria coincide throughout the unital qubit class.
\end{proof}

The proof of Theorem~\ref{thm:main} uses the following geometric lemma.
\begin{lemma}
\label{lem:weighted-ellipsoidal-average}
Let $h:S^2\rightarrow[0,\infty)$ be even and satisfy
\begin{align}
\int_{S^2}h(\hat n)d\mu(\hat n)=\frac{1}{2}.
\end{align}
Let $\{\hat u_i\}_{i\in\mc I}\subset S^2$ be a finite family such that $\vec0\in\conv\{\hat u_i:i\in\mc I\}$. Then, for arbitrary $\alpha_i>0$,
\begin{align}
\int_{S^2}
\max_{i\in\mc I}
\alpha_i\left(
\hat u_i\cdot\hat n-h(\hat n)
\right)d\mu(\hat n)
\ge0.
\label{eq:weighted-ellipsoidal-average}
\end{align}
\end{lemma}

The full proof is given in Sec.~\ref{app:weighted-ellipsoidal-average} in the appendix. For the isotropic case $h(\hat n)=\norm{D\hat{n}}=1/2$ with $D=\frac{1}{2} I$, the lemma admits the following short geometric proof.
\begin{proof}[Proof for $h(\hat n)=1/2$]
After discarding duplicate directions, let
\begin{align}
C_i:=\left\{\hat n\in S^2:\hat u_i\cdot\hat n\ge\hat u_j\cdot\hat n\ \forall j
\right\}
\end{align}
be the partitioning induced by $\{\hat u_i\}_i$. Since
$\vec0\in\conv\{\hat u_i\}_i$, one has
$\max_i\hat u_i\cdot\hat n\ge0$ for every $\hat n$, and hence
$C_i\subseteq\{\hat n:\hat u_i\cdot\hat n\ge0\}$.

Each $C_i$ is an intersection of hemispheres containing $\hat u_i$
, hence, in polar coordinates $(\theta,\varphi)$ about $\hat u_i$, we can write
$C_i=\{0\le\theta\le R_i(\varphi)\}$, where
$0\le R_i(\varphi)\le\pi/2$, and define $f_i(\hat n):=2\hat u_i\cdot\hat n-1$. Then
\begin{align*}
&\int_{C_i}f_i(\hat n)d\mu(\hat n)=\frac1{4\pi}\int\cos R_i(\varphi)\left[1-\cos R_i(\varphi)\right]d\varphi\ge0,
\end{align*}
and therefore,
\begin{align*}
&\int_{S^2}\max_{i\in\mc I}\alpha_if_i(\hat n)d\mu(\hat n)\ge\sum_i\alpha_i\int_{C_i}f_i(\hat n)d\mu(\hat n)\ge0.
\end{align*}
\end{proof}


Since Eq.~\eqref{eq:ellipsoidal-parent} reduces to the Haar-random qubit parent at $D=\frac{1}{2} I$, we have
\begin{corollary}
\label{coro:main}
The Haar-random qubit parent simulates every qubit POVM depolarized to visibility $1/2$.
\end{corollary}
Corollary~\ref{coro:main} recovers the exact Werner steering threshold. Unlike the constructive proofs of Refs.~\cite{Zhang2024, Renner2024}, the present argument uses only the dual witness inequality and applies to every unital qubit channel.

\textit{Nonunital qubit channels.--}  The dual formulation in Proposition~\ref{prop:orthogonality-normal-form} can also be used to derive a sufficient incompatibility-breaking criterion for nonunital qubit channels. Let $\mc N_{D,\vec t}$ be a qubit channel whose Heisenberg adjoint has the
Bloch representation
\begin{align}
\mc N_{D,\vec t}^{\dagger}\left(x_0\mbb{1}+\vec x\cdot\vec\sigma\right)=\left(x_0+\vec t\cdot\vec x\right)\mbb{1}+(D\vec x)\cdot\vec\sigma,
\label{eq:nonunital-heisenberg-channel}
\end{align}
where $D$ is a real $3\times3$ matrix and $\vec t\in\mbb R^3$.
\begin{proposition}
\label{prop:nonunital-sufficient}
The qubit channel $\mc N_{D,\vec t}$ is incompatibility
breaking for arbitrary POVMs if
\begin{align}
2\int_{S^2}\norm{D\hat n} d\mu(\hat n)+\norm{\vec t}\le1.
\label{eq:nonunital-sufficient-condition}
\end{align}
A parent POVM realizing this condition is of the form
\begin{align}
\Pi_{D,\vec t}(\hat n)d\mu(\hat n):=2\left[\left(h_{D,\vec t}(\hat n)+\vec t\cdot\hat n\right)\mbb{1}+(D\hat n)\cdot\vec\sigma\right]d\mu(\hat n),
\label{eq:nonunital-parent-main}
\end{align}
with $h_{D,\vec t}(\hat n):=\norm{D\hat n}+\left|\vec t\cdot\hat n\right|+\delta_{D,\vec t}$, and $\delta_{D,\vec t}=\frac{1}{2}-\frac{\norm{\vec t}}{2}-\int_{S^2}\norm{D\hat n} d\mu(\hat n)\ge 0$. The function $h_{D,\vec t}(\hat n)$ is nonnegative and even. The resulting operator density $\Pi_{D,\vec t}(\hat n)$ is positive and normalized, but is generally not centrally symmetric. 
\end{proposition}

The proof is given in
Appendix~\ref{app:nonunital-sufficient}. For
$\vec t=\vec0$, Eq.~\eqref{eq:nonunital-sufficient-condition} reduces
to the exact unital condition in
Eq.~\eqref{eq:unital-main-condition}. For $\vec t\ne\vec0$, Eq.~\eqref{eq:nonunital-sufficient-condition} is only claimed to be sufficient.


\textit{Steering criteria for two-qubit states--} We next translate the channel conditions into steering criteria.
For steering from Alice to Bob, an invertible local filter on Bob
brings every state with full-rank $\rho_B$ to the canonical form
\begin{align}
\rho_{AB}=\frac{1}{4}\left[\mbb 1\otimes\mbb 1+\vec a\cdot\vec\sigma\otimes\mbb 1+\sum_{j,k=1}^{3}T_{jk}\sigma_j\otimes\sigma_k
\right],
\label{eq:canonical-two-qubit-state}
\end{align}
without changing its steerability~\cite{Sania2014}. The
steering--joint-measurability correspondence~\cite{Uola2015,Kiukas2017} then yields the following theorem, detailed in the Supplemental Material.

\begin{theorem}
\label{thm:two-qubit-steering}
An arbitrary two-qubit state is unsteerable from Alice to Bob under all POVMs whenever the parameters of its canonical form
satisfy
\begin{align}
2\int_{S^2}\norm{T\hat n}d\mu(\hat n)+\norm{\vec a}\le1.
\label{eq:general-two-qubit-unsteerability}
\end{align}
If $\vec a=\vec0$, this condition is both necessary and sufficient.
In this case, the canonical state is locally unitarily equivalent to a Bell-diagonal state, and the condition coincides with the known unsteerability criteria under PVMs~\cite{Jevtic2015, Nguyen2016, Nguyen2019}.
\end{theorem}
\begin{example}
Consider the canonical two-qubit state $\rho^{AB}$ in
Eq.~\eqref{eq:canonical-two-qubit-state} with
\begin{align}
\vec a=\left(0,0,\frac7{50}\right),
\quad
T=\operatorname{diag}\left(\frac{13}{25},\frac{13}{25},-\frac1{20}\right).
\end{align}
It is positive semidefinite and has unit trace, while
\begin{align}
\lambda_{\min}\left([\rho^{AB}]^{T_B}\right)=\frac{95-2\sqrt{2753}}{400}<0,
\end{align}
and hence $\rho^{AB}$ is entangled by the PPT criterion. Moreover, using $\int_{S^2}\norm{T\hat n}d\mu(\hat n) \le \sqrt{\int_{S^2}\norm{T\hat n}^2d\mu(\hat n)}$,
\begin{align}
2\int_{S^2}\norm{T\hat n}d\mu(\hat n)+\norm{\vec a}
\le
\frac{\sqrt{1811}+7}{50}<1.
\end{align}
Theorem~\ref{thm:two-qubit-steering} therefore implies that
$\rho^{AB}$ is unsteerable from Alice to Bob under arbitrary POVMs. 

This example is nontrivial: Appendix~\ref{app:nontrivial-nonunital-example}
shows that $\rho^{AB}\notin\conv(\mc U_0\cup\text{SEP})$, where $\mc U_0$ denotes the set of physical canonical states with $\vec a=\vec 0$ satisfying Eq.~\eqref{eq:general-two-qubit-unsteerability}, and $\text{SEP}$ is the set of all separable states. 
\end{example}

\textit{Higher dimensions.--} The same dual framework applies to depolarizing noise in arbitrary dimension, the other natural generalization of the Werner problem. Let
\begin{subequations}
\begin{align}
\Delta_r(X)&=rX+(1-r)\frac{\tr X}{d}\mbb{1}, \\
\Pi_{\rm Haar}(d\psi)&=d\op{\psi}{\psi}d\mu_H(\psi).
\end{align}
\end{subequations}
The Haar parent simulates all depolarized PVMs up to $r_d=\frac{H_d-1}{d-1}$ with $H_d=\sum_{k=1}^d\frac{1}{k}$~\cite{Wiseman2007,Almeida2007,Ioannou2022}. Proposition~\ref{prop:orthogonality-normal-form} then gives the following exact reduction.
\begin{corollary}
\label{coro:haar-positive-form}
The Haar parent $\Pi_{\rm Haar}$ simulates every POVM
$\{\Delta_{r_d}(M_a)\}_a$ if and only if
\begin{align}
\int\max_a\left[\frac{1-r_d}{d}\tr(P_a)-\langle\psi|P_a|\psi\rangle\right]d\mu_H(\psi)\ge0
\label{eq:haar-positive-form}
\end{align}
for every finite family $P_a\ge\mbb{0}$ for which there exists a POVM
$\{Q_a\}_a$ satisfying $P_aQ_a=\mbb{0}$ for all $a$.
\end{corollary}

For $d=2$, outcome-wise orthogonality forces $P_a$ to be rank one, reducing the problem to a special case of Eq.~\eqref{eq:weighted-ellipsoidal-average}; the spherical averaging lemma above is then the step special to qubits. For $d>2$, the operators $P_a$ retain nontrivial spectral structure, and the resulting geometric inequality remains open.

\textit{Conclusion--} We have proved that PVM- and POVM-incompatibility breaking coincide for every unital qubit channel. Starting from the parent POVM associated with the known PVM criterion, support-function duality shows that the same parent simulates all noisy POVMs at the same noise threshold. For nonunital qubit channels, a generally asymmetric parent yields an explicit sufficient incompatibility-breaking condition. Through the steering--joint-measurability correspondence, these results give the exact POVM-steering criterion for two-qubit states with maximally mixed marginals and a sufficient unsteerability criterion for arbitrary two-qubit states.

The dual formulation also makes the remaining open problems more explicit. For nonunital qubit channels, determining the exact boundary requires a better choice of parent POVM. In higher dimensions, the Haar parent reaches the exact PVM threshold, while Proposition~\ref{prop:orthogonality-normal-form} reduces the question of whether this same parent also simulates all POVMs to the geometric inequality in Eq.~\eqref{eq:haar-positive-form}.



\textit{Acknowledgments--} Y.Z. thanks Eric Chitambar for helpful discussions. Y.Z. acknowledges the support of the Institute for Quantum Computing, the University of Waterloo, and the Perimeter Institute. Research at the Institute for Quantum Computing is supported by Innovation, Science and Economic Development Canada. Research at Perimeter Institute is supported in part by the Government of Canada through the Department of Innovation, Science and Economic Development, and by the Province of Ontario through the Ministry of Colleges and Universities. 

The author acknowledges the use of OpenAI Codex and ChatGPT for exploratory numerical examples and to assist in refining the proof and improving readability and presentation. In particular, OpenAI ChatGPT (GPT-5.5) helped in generalizing Lemma~\ref{lem:weighted-ellipsoidal-average}, initially formulated for $h(\hat{n})=\norm{D\hat{n}}$ to arbitrary even $h(\hat{n})$. The author takes full responsibility for the originality and mathematical correctness of the results.


\bibliography{ref}

\begin{thebibliography}{48}%
\makeatletter
\providecommand \@ifxundefined [1]{%
 \@ifx{#1\undefined}
}%
\providecommand \@ifnum [1]{%
 \ifnum #1\expandafter \@firstoftwo
 \else \expandafter \@secondoftwo
 \fi
}%
\providecommand \@ifx [1]{%
 \ifx #1\expandafter \@firstoftwo
 \else \expandafter \@secondoftwo
 \fi
}%
\providecommand \natexlab [1]{#1}%
\providecommand \enquote  [1]{``#1''}%
\providecommand \bibnamefont  [1]{#1}%
\providecommand \bibfnamefont [1]{#1}%
\providecommand \citenamefont [1]{#1}%
\providecommand \href@noop [0]{\@secondoftwo}%
\providecommand \href [0]{\begingroup \@sanitize@url \@href}%
\providecommand \@href[1]{\@@startlink{#1}\@@href}%
\providecommand \@@href[1]{\endgroup#1\@@endlink}%
\providecommand \@sanitize@url [0]{\catcode `\\12\catcode `\$12\catcode `\&12\catcode `\#12\catcode `\^12\catcode `\_12\catcode `\%12\relax}%
\providecommand \@@startlink[1]{}%
\providecommand \@@endlink[0]{}%
\providecommand \url  [0]{\begingroup\@sanitize@url \@url }%
\providecommand \@url [1]{\endgroup\@href {#1}{\urlprefix }}%
\providecommand \urlprefix  [0]{URL }%
\providecommand \Eprint [0]{\href }%
\providecommand \doibase [0]{https://doi.org/}%
\providecommand \selectlanguage [0]{\@gobble}%
\providecommand \bibinfo  [0]{\@secondoftwo}%
\providecommand \bibfield  [0]{\@secondoftwo}%
\providecommand \translation [1]{[#1]}%
\providecommand \BibitemOpen [0]{}%
\providecommand \bibitemStop [0]{}%
\providecommand \bibitemNoStop [0]{.\EOS\space}%
\providecommand \EOS [0]{\spacefactor3000\relax}%
\providecommand \BibitemShut  [1]{\csname bibitem#1\endcsname}%
\let\auto@bib@innerbib\@empty
\bibitem [{\citenamefont {Heinosaari}\ \emph {et~al.}(2016)\citenamefont {Heinosaari}, \citenamefont {Miyadera},\ and\ \citenamefont {Ziman}}]{Heinosaari2016}%
  \BibitemOpen
  \bibfield  {author} {\bibinfo {author} {\bibfnamefont {T.}~\bibnamefont {Heinosaari}}, \bibinfo {author} {\bibfnamefont {T.}~\bibnamefont {Miyadera}},\ and\ \bibinfo {author} {\bibfnamefont {M.}~\bibnamefont {Ziman}},\ }\bibfield  {title} {\bibinfo {title} {An invitation to quantum incompatibility},\ }\href {https://doi.org/10.1088/1751-8113/49/12/123001} {\bibfield  {journal} {\bibinfo  {journal} {Journal of Physics A: Mathematical and Theoretical}\ }\textbf {\bibinfo {volume} {49}},\ \bibinfo {pages} {123001} (\bibinfo {year} {2016})}\BibitemShut {NoStop}%
\bibitem [{\citenamefont {G{\"u}hne}\ \emph {et~al.}(2023)\citenamefont {G{\"u}hne}, \citenamefont {Haapasalo}, \citenamefont {Kraft}, \citenamefont {Pellonp{\"a}{\"a}},\ and\ \citenamefont {Uola}}]{Guhne2023}%
  \BibitemOpen
  \bibfield  {author} {\bibinfo {author} {\bibfnamefont {O.}~\bibnamefont {G{\"u}hne}}, \bibinfo {author} {\bibfnamefont {E.}~\bibnamefont {Haapasalo}}, \bibinfo {author} {\bibfnamefont {T.}~\bibnamefont {Kraft}}, \bibinfo {author} {\bibfnamefont {J.-P.}\ \bibnamefont {Pellonp{\"a}{\"a}}},\ and\ \bibinfo {author} {\bibfnamefont {R.}~\bibnamefont {Uola}},\ }\bibfield  {title} {\bibinfo {title} {Colloquium: Incompatible measurements in quantum information science},\ }\href {https://doi.org/10.1103/RevModPhys.95.011003} {\bibfield  {journal} {\bibinfo  {journal} {Rev. Mod. Phys.}\ }\textbf {\bibinfo {volume} {95}},\ \bibinfo {pages} {011003} (\bibinfo {year} {2023})}\BibitemShut {NoStop}%
\bibitem [{\citenamefont {Quintino}\ \emph {et~al.}(2015)\citenamefont {Quintino}, \citenamefont {V\'ertesi}, \citenamefont {Cavalcanti}, \citenamefont {Augusiak}, \citenamefont {Demianowicz}, \citenamefont {Ac\'{\i}n},\ and\ \citenamefont {Brunner}}]{Quintino2015}%
  \BibitemOpen
  \bibfield  {author} {\bibinfo {author} {\bibfnamefont {M.~T.}\ \bibnamefont {Quintino}}, \bibinfo {author} {\bibfnamefont {T.}~\bibnamefont {V\'ertesi}}, \bibinfo {author} {\bibfnamefont {D.}~\bibnamefont {Cavalcanti}}, \bibinfo {author} {\bibfnamefont {R.}~\bibnamefont {Augusiak}}, \bibinfo {author} {\bibfnamefont {M.}~\bibnamefont {Demianowicz}}, \bibinfo {author} {\bibfnamefont {A.}~\bibnamefont {Ac\'{\i}n}},\ and\ \bibinfo {author} {\bibfnamefont {N.}~\bibnamefont {Brunner}},\ }\bibfield  {title} {\bibinfo {title} {Inequivalence of entanglement, steering, and bell nonlocality for general measurements},\ }\href {https://doi.org/10.1103/PhysRevA.92.032107} {\bibfield  {journal} {\bibinfo  {journal} {Phys. Rev. A}\ }\textbf {\bibinfo {volume} {92}},\ \bibinfo {pages} {032107} (\bibinfo {year} {2015})}\BibitemShut {NoStop}%
\bibitem [{\citenamefont {Cavalcanti}\ and\ \citenamefont {Skrzypczyk}(2016{\natexlab{a}})}]{CavalcantiSkrzypczyk2016}%
  \BibitemOpen
  \bibfield  {author} {\bibinfo {author} {\bibfnamefont {D.}~\bibnamefont {Cavalcanti}}\ and\ \bibinfo {author} {\bibfnamefont {P.}~\bibnamefont {Skrzypczyk}},\ }\bibfield  {title} {\bibinfo {title} {Quantitative relations between measurement incompatibility, quantum steering, and nonlocality},\ }\href {https://doi.org/10.1103/PhysRevA.93.052112} {\bibfield  {journal} {\bibinfo  {journal} {Phys. Rev. A}\ }\textbf {\bibinfo {volume} {93}},\ \bibinfo {pages} {052112} (\bibinfo {year} {2016}{\natexlab{a}})}\BibitemShut {NoStop}%
\bibitem [{\citenamefont {Hirsch}\ \emph {et~al.}(2018)\citenamefont {Hirsch}, \citenamefont {Quintino},\ and\ \citenamefont {Brunner}}]{Hirsch2018}%
  \BibitemOpen
  \bibfield  {author} {\bibinfo {author} {\bibfnamefont {F.}~\bibnamefont {Hirsch}}, \bibinfo {author} {\bibfnamefont {M.~T.}\ \bibnamefont {Quintino}},\ and\ \bibinfo {author} {\bibfnamefont {N.}~\bibnamefont {Brunner}},\ }\bibfield  {title} {\bibinfo {title} {Quantum measurement incompatibility does not imply bell nonlocality},\ }\href {https://doi.org/10.1103/PhysRevA.97.012129} {\bibfield  {journal} {\bibinfo  {journal} {Phys. Rev. A}\ }\textbf {\bibinfo {volume} {97}},\ \bibinfo {pages} {012129} (\bibinfo {year} {2018})}\BibitemShut {NoStop}%
\bibitem [{\citenamefont {Bene}\ and\ \citenamefont {V{\'e}rtesi}(2018)}]{BeneVertesi2018}%
  \BibitemOpen
  \bibfield  {author} {\bibinfo {author} {\bibfnamefont {E.}~\bibnamefont {Bene}}\ and\ \bibinfo {author} {\bibfnamefont {T.}~\bibnamefont {V{\'e}rtesi}},\ }\bibfield  {title} {\bibinfo {title} {Measurement incompatibility does not give rise to bell violation in general},\ }\href {https://doi.org/10.1088/1367-2630/aa9ca3} {\bibfield  {journal} {\bibinfo  {journal} {New Journal of Physics}\ }\textbf {\bibinfo {volume} {20}},\ \bibinfo {pages} {013021} (\bibinfo {year} {2018})}\BibitemShut {NoStop}%
\bibitem [{\citenamefont {Wolf}\ \emph {et~al.}(2009)\citenamefont {Wolf}, \citenamefont {Perez-Garcia},\ and\ \citenamefont {Fernandez}}]{Wolf2009}%
  \BibitemOpen
  \bibfield  {author} {\bibinfo {author} {\bibfnamefont {M.~M.}\ \bibnamefont {Wolf}}, \bibinfo {author} {\bibfnamefont {D.}~\bibnamefont {Perez-Garcia}},\ and\ \bibinfo {author} {\bibfnamefont {C.}~\bibnamefont {Fernandez}},\ }\bibfield  {title} {\bibinfo {title} {Measurements incompatible in quantum theory cannot be measured jointly in any other no-signaling theory},\ }\href {https://doi.org/10.1103/PhysRevLett.103.230402} {\bibfield  {journal} {\bibinfo  {journal} {Phys. Rev. Lett.}\ }\textbf {\bibinfo {volume} {103}},\ \bibinfo {pages} {230402} (\bibinfo {year} {2009})}\BibitemShut {NoStop}%
\bibitem [{\citenamefont {Pl{\'a}vala}\ \emph {et~al.}(2025)\citenamefont {Pl{\'a}vala}, \citenamefont {G{\"u}hne},\ and\ \citenamefont {Quintino}}]{Plavala2025}%
  \BibitemOpen
  \bibfield  {author} {\bibinfo {author} {\bibfnamefont {M.}~\bibnamefont {Pl{\'a}vala}}, \bibinfo {author} {\bibfnamefont {O.}~\bibnamefont {G{\"u}hne}},\ and\ \bibinfo {author} {\bibfnamefont {M.~T.}\ \bibnamefont {Quintino}},\ }\bibfield  {title} {\bibinfo {title} {All incompatible measurements on qubits lead to multiparticle bell nonlocality},\ }\href {https://doi.org/10.1103/PhysRevLett.134.200201} {\bibfield  {journal} {\bibinfo  {journal} {Phys. Rev. Lett.}\ }\textbf {\bibinfo {volume} {134}},\ \bibinfo {pages} {200201} (\bibinfo {year} {2025})}\BibitemShut {NoStop}%
\bibitem [{\citenamefont {Uola}\ \emph {et~al.}(2020)\citenamefont {Uola}, \citenamefont {Costa}, \citenamefont {Nguyen},\ and\ \citenamefont {G\"uhne}}]{Uola2020}%
  \BibitemOpen
  \bibfield  {author} {\bibinfo {author} {\bibfnamefont {R.}~\bibnamefont {Uola}}, \bibinfo {author} {\bibfnamefont {A.~C.~S.}\ \bibnamefont {Costa}}, \bibinfo {author} {\bibfnamefont {H.~C.}\ \bibnamefont {Nguyen}},\ and\ \bibinfo {author} {\bibfnamefont {O.}~\bibnamefont {G\"uhne}},\ }\bibfield  {title} {\bibinfo {title} {Quantum steering},\ }\href {https://doi.org/10.1103/RevModPhys.92.015001} {\bibfield  {journal} {\bibinfo  {journal} {Rev. Mod. Phys.}\ }\textbf {\bibinfo {volume} {92}},\ \bibinfo {pages} {015001} (\bibinfo {year} {2020})}\BibitemShut {NoStop}%
\bibitem [{\citenamefont {Quintino}\ \emph {et~al.}(2014)\citenamefont {Quintino}, \citenamefont {V\'ertesi},\ and\ \citenamefont {Brunner}}]{Quintino2014}%
  \BibitemOpen
  \bibfield  {author} {\bibinfo {author} {\bibfnamefont {M.~T.}\ \bibnamefont {Quintino}}, \bibinfo {author} {\bibfnamefont {T.}~\bibnamefont {V\'ertesi}},\ and\ \bibinfo {author} {\bibfnamefont {N.}~\bibnamefont {Brunner}},\ }\bibfield  {title} {\bibinfo {title} {Joint measurability, einstein-podolsky-rosen steering, and bell nonlocality},\ }\href {https://doi.org/10.1103/PhysRevLett.113.160402} {\bibfield  {journal} {\bibinfo  {journal} {Phys. Rev. Lett.}\ }\textbf {\bibinfo {volume} {113}},\ \bibinfo {pages} {160402} (\bibinfo {year} {2014})}\BibitemShut {NoStop}%
\bibitem [{\citenamefont {Uola}\ \emph {et~al.}(2014)\citenamefont {Uola}, \citenamefont {Moroder},\ and\ \citenamefont {G\"uhne}}]{Uola2014}%
  \BibitemOpen
  \bibfield  {author} {\bibinfo {author} {\bibfnamefont {R.}~\bibnamefont {Uola}}, \bibinfo {author} {\bibfnamefont {T.}~\bibnamefont {Moroder}},\ and\ \bibinfo {author} {\bibfnamefont {O.}~\bibnamefont {G\"uhne}},\ }\bibfield  {title} {\bibinfo {title} {Joint measurability of generalized measurements implies classicality},\ }\href {https://doi.org/10.1103/PhysRevLett.113.160403} {\bibfield  {journal} {\bibinfo  {journal} {Phys. Rev. Lett.}\ }\textbf {\bibinfo {volume} {113}},\ \bibinfo {pages} {160403} (\bibinfo {year} {2014})}\BibitemShut {NoStop}%
\bibitem [{\citenamefont {Uola}\ \emph {et~al.}(2015)\citenamefont {Uola}, \citenamefont {Budroni}, \citenamefont {G\"uhne},\ and\ \citenamefont {Pellonp\"a\"a}}]{Uola2015}%
  \BibitemOpen
  \bibfield  {author} {\bibinfo {author} {\bibfnamefont {R.}~\bibnamefont {Uola}}, \bibinfo {author} {\bibfnamefont {C.}~\bibnamefont {Budroni}}, \bibinfo {author} {\bibfnamefont {O.}~\bibnamefont {G\"uhne}},\ and\ \bibinfo {author} {\bibfnamefont {J.-P.}\ \bibnamefont {Pellonp\"a\"a}},\ }\bibfield  {title} {\bibinfo {title} {One-to-one mapping between steering and joint measurability problems},\ }\href {https://doi.org/10.1103/PhysRevLett.115.230402} {\bibfield  {journal} {\bibinfo  {journal} {Phys. Rev. Lett.}\ }\textbf {\bibinfo {volume} {115}},\ \bibinfo {pages} {230402} (\bibinfo {year} {2015})}\BibitemShut {NoStop}%
\bibitem [{\citenamefont {Cavalcanti}\ and\ \citenamefont {Skrzypczyk}(2016{\natexlab{b}})}]{Cavalcanti2016}%
  \BibitemOpen
  \bibfield  {author} {\bibinfo {author} {\bibfnamefont {D.}~\bibnamefont {Cavalcanti}}\ and\ \bibinfo {author} {\bibfnamefont {P.}~\bibnamefont {Skrzypczyk}},\ }\bibfield  {title} {\bibinfo {title} {Quantum steering: a review with focus on semidefinite programming},\ }\href {https://doi.org/10.1088/1361-6633/80/2/024001} {\bibfield  {journal} {\bibinfo  {journal} {Reports on Progress in Physics}\ }\textbf {\bibinfo {volume} {80}},\ \bibinfo {pages} {024001} (\bibinfo {year} {2016}{\natexlab{b}})}\BibitemShut {NoStop}%
\bibitem [{\citenamefont {Kiukas}\ \emph {et~al.}(2017)\citenamefont {Kiukas}, \citenamefont {Budroni}, \citenamefont {Uola},\ and\ \citenamefont {Pellonp\"a\"a}}]{Kiukas2017}%
  \BibitemOpen
  \bibfield  {author} {\bibinfo {author} {\bibfnamefont {J.}~\bibnamefont {Kiukas}}, \bibinfo {author} {\bibfnamefont {C.}~\bibnamefont {Budroni}}, \bibinfo {author} {\bibfnamefont {R.}~\bibnamefont {Uola}},\ and\ \bibinfo {author} {\bibfnamefont {J.-P.}\ \bibnamefont {Pellonp\"a\"a}},\ }\bibfield  {title} {\bibinfo {title} {Continuous-variable steering and incompatibility via state-channel duality},\ }\href {https://doi.org/10.1103/PhysRevA.96.042331} {\bibfield  {journal} {\bibinfo  {journal} {Phys. Rev. A}\ }\textbf {\bibinfo {volume} {96}},\ \bibinfo {pages} {042331} (\bibinfo {year} {2017})}\BibitemShut {NoStop}%
\bibitem [{\citenamefont {Spekkens}(2005)}]{Spekkens2005}%
  \BibitemOpen
  \bibfield  {author} {\bibinfo {author} {\bibfnamefont {R.~W.}\ \bibnamefont {Spekkens}},\ }\bibfield  {title} {\bibinfo {title} {Contextuality for preparations, transformations, and unsharp measurements},\ }\href {https://doi.org/10.1103/PhysRevA.71.052108} {\bibfield  {journal} {\bibinfo  {journal} {Phys. Rev. A}\ }\textbf {\bibinfo {volume} {71}},\ \bibinfo {pages} {052108} (\bibinfo {year} {2005})}\BibitemShut {NoStop}%
\bibitem [{\citenamefont {Tavakoli}\ and\ \citenamefont {Uola}(2020)}]{TavakoliUola2020}%
  \BibitemOpen
  \bibfield  {author} {\bibinfo {author} {\bibfnamefont {A.}~\bibnamefont {Tavakoli}}\ and\ \bibinfo {author} {\bibfnamefont {R.}~\bibnamefont {Uola}},\ }\bibfield  {title} {\bibinfo {title} {Measurement incompatibility and steering are necessary and sufficient for operational contextuality},\ }\href {https://doi.org/10.1103/PhysRevResearch.2.013011} {\bibfield  {journal} {\bibinfo  {journal} {Phys. Rev. Research}\ }\textbf {\bibinfo {volume} {2}},\ \bibinfo {pages} {013011} (\bibinfo {year} {2020})}\BibitemShut {NoStop}%
\bibitem [{\citenamefont {Selby}\ \emph {et~al.}(2023)\citenamefont {Selby}, \citenamefont {Schmid}, \citenamefont {Wolfe}, \citenamefont {Sainz}, \citenamefont {Kunjwal},\ and\ \citenamefont {Spekkens}}]{Selby2023}%
  \BibitemOpen
  \bibfield  {author} {\bibinfo {author} {\bibfnamefont {J.~H.}\ \bibnamefont {Selby}}, \bibinfo {author} {\bibfnamefont {D.}~\bibnamefont {Schmid}}, \bibinfo {author} {\bibfnamefont {E.}~\bibnamefont {Wolfe}}, \bibinfo {author} {\bibfnamefont {A.~B.}\ \bibnamefont {Sainz}}, \bibinfo {author} {\bibfnamefont {R.}~\bibnamefont {Kunjwal}},\ and\ \bibinfo {author} {\bibfnamefont {R.~W.}\ \bibnamefont {Spekkens}},\ }\bibfield  {title} {\bibinfo {title} {Contextuality without incompatibility},\ }\href {https://doi.org/10.1103/PhysRevLett.130.230201} {\bibfield  {journal} {\bibinfo  {journal} {Phys. Rev. Lett.}\ }\textbf {\bibinfo {volume} {130}},\ \bibinfo {pages} {230201} (\bibinfo {year} {2023})}\BibitemShut {NoStop}%
\bibitem [{\citenamefont {Zhang}\ \emph {et~al.}(2026)\citenamefont {Zhang}, \citenamefont {Schmid}, \citenamefont {Y{\=i}ng},\ and\ \citenamefont {Spekkens}}]{ZhangSchmidYingSpekkens2026}%
  \BibitemOpen
  \bibfield  {author} {\bibinfo {author} {\bibfnamefont {Y.}~\bibnamefont {Zhang}}, \bibinfo {author} {\bibfnamefont {D.}~\bibnamefont {Schmid}}, \bibinfo {author} {\bibfnamefont {Y.}~\bibnamefont {Y{\=i}ng}},\ and\ \bibinfo {author} {\bibfnamefont {R.~W.}\ \bibnamefont {Spekkens}},\ }\bibfield  {title} {\bibinfo {title} {Reassessing the boundary between classical and nonclassical for individual quantum processes},\ }\href {https://doi.org/10.1103/vqfz-wzjg} {\bibfield  {journal} {\bibinfo  {journal} {Phys. Rev. X}\ }\textbf {\bibinfo {volume} {16}},\ \bibinfo {pages} {021050} (\bibinfo {year} {2026})}\BibitemShut {NoStop}%
\bibitem [{\citenamefont {Zhang}\ \emph {et~al.}(2025{\natexlab{a}})\citenamefont {Zhang}, \citenamefont {Y{\=i}ng},\ and\ \citenamefont {Schmid}}]{ZhangYingSchmid2025}%
  \BibitemOpen
  \bibfield  {author} {\bibinfo {author} {\bibfnamefont {Y.}~\bibnamefont {Zhang}}, \bibinfo {author} {\bibfnamefont {Y.}~\bibnamefont {Y{\=i}ng}},\ and\ \bibinfo {author} {\bibfnamefont {D.}~\bibnamefont {Schmid}},\ }\href {https://doi.org/10.48550/arXiv.2504.02944} {\bibinfo {title} {Quantifiers and witnesses for the nonclassicality of measurements and of states}} (\bibinfo {year} {2025}{\natexlab{a}}),\ \Eprint {https://arxiv.org/abs/2504.02944} {arXiv:2504.02944 [quant-ph]} \BibitemShut {NoStop}%
\bibitem [{\citenamefont {Carmeli}\ \emph {et~al.}(2019)\citenamefont {Carmeli}, \citenamefont {Heinosaari},\ and\ \citenamefont {Toigo}}]{Carmeli2019a}%
  \BibitemOpen
  \bibfield  {author} {\bibinfo {author} {\bibfnamefont {C.}~\bibnamefont {Carmeli}}, \bibinfo {author} {\bibfnamefont {T.}~\bibnamefont {Heinosaari}},\ and\ \bibinfo {author} {\bibfnamefont {A.}~\bibnamefont {Toigo}},\ }\bibfield  {title} {\bibinfo {title} {Quantum incompatibility witnesses},\ }\href {https://doi.org/10.1103/PhysRevLett.122.130402} {\bibfield  {journal} {\bibinfo  {journal} {Phys. Rev. Lett.}\ }\textbf {\bibinfo {volume} {122}},\ \bibinfo {pages} {130402} (\bibinfo {year} {2019})}\BibitemShut {NoStop}%
\bibitem [{\citenamefont {Skrzypczyk}\ \emph {et~al.}(2019)\citenamefont {Skrzypczyk}, \citenamefont {\ifmmode \check{S}\else \v{S}\fi{}upi\ifmmode~\acute{c}\else \'{c}\fi{}},\ and\ \citenamefont {Cavalcanti}}]{Skrzypczyk2019}%
  \BibitemOpen
  \bibfield  {author} {\bibinfo {author} {\bibfnamefont {P.}~\bibnamefont {Skrzypczyk}}, \bibinfo {author} {\bibfnamefont {I.}~\bibnamefont {\ifmmode \check{S}\else \v{S}\fi{}upi\ifmmode~\acute{c}\else \'{c}\fi{}}},\ and\ \bibinfo {author} {\bibfnamefont {D.}~\bibnamefont {Cavalcanti}},\ }\bibfield  {title} {\bibinfo {title} {All sets of incompatible measurements give an advantage in quantum state discrimination},\ }\href {https://doi.org/10.1103/PhysRevLett.122.130403} {\bibfield  {journal} {\bibinfo  {journal} {Phys. Rev. Lett.}\ }\textbf {\bibinfo {volume} {122}},\ \bibinfo {pages} {130403} (\bibinfo {year} {2019})}\BibitemShut {NoStop}%
\bibitem [{\citenamefont {Buscemi}\ \emph {et~al.}(2020)\citenamefont {Buscemi}, \citenamefont {Chitambar},\ and\ \citenamefont {Zhou}}]{Buscemi2020}%
  \BibitemOpen
  \bibfield  {author} {\bibinfo {author} {\bibfnamefont {F.}~\bibnamefont {Buscemi}}, \bibinfo {author} {\bibfnamefont {E.}~\bibnamefont {Chitambar}},\ and\ \bibinfo {author} {\bibfnamefont {W.}~\bibnamefont {Zhou}},\ }\bibfield  {title} {\bibinfo {title} {Complete resource theory of quantum incompatibility as quantum programmability},\ }\href {https://doi.org/10.1103/PhysRevLett.124.120401} {\bibfield  {journal} {\bibinfo  {journal} {Phys. Rev. Lett.}\ }\textbf {\bibinfo {volume} {124}},\ \bibinfo {pages} {120401} (\bibinfo {year} {2020})}\BibitemShut {NoStop}%
\bibitem [{\citenamefont {Heinosaari}\ \emph {et~al.}(2015)\citenamefont {Heinosaari}, \citenamefont {Kiukas}, \citenamefont {Reitzner},\ and\ \citenamefont {Schultz}}]{Heinosaari2015}%
  \BibitemOpen
  \bibfield  {author} {\bibinfo {author} {\bibfnamefont {T.}~\bibnamefont {Heinosaari}}, \bibinfo {author} {\bibfnamefont {J.}~\bibnamefont {Kiukas}}, \bibinfo {author} {\bibfnamefont {D.}~\bibnamefont {Reitzner}},\ and\ \bibinfo {author} {\bibfnamefont {J.}~\bibnamefont {Schultz}},\ }\bibfield  {title} {\bibinfo {title} {Incompatibility breaking quantum channels},\ }\href {https://doi.org/10.1088/1751-8113/48/43/435301} {\bibfield  {journal} {\bibinfo  {journal} {Journal of Physics A: Mathematical and Theoretical}\ }\textbf {\bibinfo {volume} {48}},\ \bibinfo {pages} {435301} (\bibinfo {year} {2015})}\BibitemShut {NoStop}%
\bibitem [{\citenamefont {Zhang}\ and\ \citenamefont {Chitambar}(2026)}]{zhang2026parallel}%
  \BibitemOpen
  \bibfield  {author} {\bibinfo {author} {\bibfnamefont {Y.}~\bibnamefont {Zhang}}\ and\ \bibinfo {author} {\bibfnamefont {E.}~\bibnamefont {Chitambar}},\ }\bibfield  {title} {\bibinfo {title} {Destroying and preserving measurement incompatibility over a quantum channel}} (\bibinfo {year} {2026})\BibitemShut {NoStop}%
\bibitem [{\citenamefont {Peres}(1988)}]{Peres1988}%
  \BibitemOpen
  \bibfield  {author} {\bibinfo {author} {\bibfnamefont {A.}~\bibnamefont {Peres}},\ }\bibfield  {title} {\bibinfo {title} {How to differentiate between non-orthogonal states},\ }\href {https://doi.org/10.1016/0375-9601(88)91034-1} {\bibfield  {journal} {\bibinfo  {journal} {Physics Letters A}\ }\textbf {\bibinfo {volume} {128}},\ \bibinfo {pages} {19} (\bibinfo {year} {1988})}\BibitemShut {NoStop}%
\bibitem [{\citenamefont {Oszmaniec}\ and\ \citenamefont {Biswas}(2019)}]{OszmaniecBiswas2019}%
  \BibitemOpen
  \bibfield  {author} {\bibinfo {author} {\bibfnamefont {M.}~\bibnamefont {Oszmaniec}}\ and\ \bibinfo {author} {\bibfnamefont {T.}~\bibnamefont {Biswas}},\ }\bibfield  {title} {\bibinfo {title} {Operational relevance of resource theories of quantum measurements},\ }\href {https://doi.org/10.22331/q-2019-04-26-133} {\bibfield  {journal} {\bibinfo  {journal} {Quantum}\ }\textbf {\bibinfo {volume} {3}},\ \bibinfo {pages} {133} (\bibinfo {year} {2019})},\ \Eprint {https://arxiv.org/abs/1901.08566} {arXiv:1901.08566 [quant-ph]} \BibitemShut {NoStop}%
\bibitem [{\citenamefont {Massar}\ and\ \citenamefont {Popescu}(1995)}]{MassarPopescu1995}%
  \BibitemOpen
  \bibfield  {author} {\bibinfo {author} {\bibfnamefont {S.}~\bibnamefont {Massar}}\ and\ \bibinfo {author} {\bibfnamefont {S.}~\bibnamefont {Popescu}},\ }\bibfield  {title} {\bibinfo {title} {Optimal extraction of information from finite quantum ensembles},\ }\href {https://doi.org/10.1103/PhysRevLett.74.1259} {\bibfield  {journal} {\bibinfo  {journal} {Physical Review Letters}\ }\textbf {\bibinfo {volume} {74}},\ \bibinfo {pages} {1259} (\bibinfo {year} {1995})}\BibitemShut {NoStop}%
\bibitem [{\citenamefont {V{\'e}rtesi}\ and\ \citenamefont {Bene}(2010)}]{VertesiBene2010}%
  \BibitemOpen
  \bibfield  {author} {\bibinfo {author} {\bibfnamefont {T.}~\bibnamefont {V{\'e}rtesi}}\ and\ \bibinfo {author} {\bibfnamefont {E.}~\bibnamefont {Bene}},\ }\bibfield  {title} {\bibinfo {title} {Two-qubit bell inequality for which positive operator-valued measurements are relevant},\ }\href {https://doi.org/10.1103/PhysRevA.82.062115} {\bibfield  {journal} {\bibinfo  {journal} {Physical Review A}\ }\textbf {\bibinfo {volume} {82}},\ \bibinfo {pages} {062115} (\bibinfo {year} {2010})},\ \Eprint {https://arxiv.org/abs/1007.2578} {arXiv:1007.2578 [quant-ph]} \BibitemShut {NoStop}%
\bibitem [{\citenamefont {Werner}(2014)}]{Werner2014}%
  \BibitemOpen
  \bibfield  {author} {\bibinfo {author} {\bibfnamefont {R.~F.}\ \bibnamefont {Werner}},\ }\bibfield  {title} {\bibinfo {title} {Steering, or maybe why einstein did not go all the way to bell's argument},\ }\href {https://doi.org/10.1088/1751-8113/47/42/424008} {\bibfield  {journal} {\bibinfo  {journal} {Journal of Physics A: Mathematical and Theoretical}\ }\textbf {\bibinfo {volume} {47}},\ \bibinfo {pages} {424008} (\bibinfo {year} {2014})}\BibitemShut {NoStop}%
\bibitem [{\citenamefont {Werner}(1989)}]{Werner1989}%
  \BibitemOpen
  \bibfield  {author} {\bibinfo {author} {\bibfnamefont {R.~F.}\ \bibnamefont {Werner}},\ }\bibfield  {title} {\bibinfo {title} {Quantum states with einstein-podolsky-rosen correlations admitting a hidden-variable model},\ }\href {https://doi.org/10.1103/PhysRevA.40.4277} {\bibfield  {journal} {\bibinfo  {journal} {Phys. Rev. A}\ }\textbf {\bibinfo {volume} {40}},\ \bibinfo {pages} {4277} (\bibinfo {year} {1989})}\BibitemShut {NoStop}%
\bibitem [{\citenamefont {Barrett}(2002)}]{Barrett2002}%
  \BibitemOpen
  \bibfield  {author} {\bibinfo {author} {\bibfnamefont {J.}~\bibnamefont {Barrett}},\ }\bibfield  {title} {\bibinfo {title} {Nonsequential positive-operator-valued measurements on entangled mixed states do not always violate a bell inequality},\ }\href {https://doi.org/10.1103/PhysRevA.65.042302} {\bibfield  {journal} {\bibinfo  {journal} {Phys. Rev. A}\ }\textbf {\bibinfo {volume} {65}},\ \bibinfo {pages} {042302} (\bibinfo {year} {2002})}\BibitemShut {NoStop}%
\bibitem [{\citenamefont {Zhang}\ and\ \citenamefont {Chitambar}(2024)}]{Zhang2024}%
  \BibitemOpen
  \bibfield  {author} {\bibinfo {author} {\bibfnamefont {Y.}~\bibnamefont {Zhang}}\ and\ \bibinfo {author} {\bibfnamefont {E.}~\bibnamefont {Chitambar}},\ }\bibfield  {title} {\bibinfo {title} {Exact steering bound for two-qubit werner states},\ }\href {https://doi.org/10.1103/PhysRevLett.132.250201} {\bibfield  {journal} {\bibinfo  {journal} {Phys. Rev. Lett.}\ }\textbf {\bibinfo {volume} {132}},\ \bibinfo {pages} {250201} (\bibinfo {year} {2024})}\BibitemShut {NoStop}%
\bibitem [{\citenamefont {Renner}(2024)}]{Renner2024}%
  \BibitemOpen
  \bibfield  {author} {\bibinfo {author} {\bibfnamefont {M.~J.}\ \bibnamefont {Renner}},\ }\bibfield  {title} {\bibinfo {title} {Compatibility of generalized noisy qubit measurements},\ }\href {https://doi.org/10.1103/PhysRevLett.132.250202} {\bibfield  {journal} {\bibinfo  {journal} {Phys. Rev. Lett.}\ }\textbf {\bibinfo {volume} {132}},\ \bibinfo {pages} {250202} (\bibinfo {year} {2024})}\BibitemShut {NoStop}%
\bibitem [{\citenamefont {Nguyen}\ \emph {et~al.}(2018)\citenamefont {Nguyen}, \citenamefont {Milne}, \citenamefont {Vu},\ and\ \citenamefont {Jevtic}}]{Nguyen2018}%
  \BibitemOpen
  \bibfield  {author} {\bibinfo {author} {\bibfnamefont {H.~C.}\ \bibnamefont {Nguyen}}, \bibinfo {author} {\bibfnamefont {A.}~\bibnamefont {Milne}}, \bibinfo {author} {\bibfnamefont {T.}~\bibnamefont {Vu}},\ and\ \bibinfo {author} {\bibfnamefont {S.}~\bibnamefont {Jevtic}},\ }\bibfield  {title} {\bibinfo {title} {Quantum steering with positive operator valued measures},\ }\href {https://doi.org/10.1088/1751-8121/aad115} {\bibfield  {journal} {\bibinfo  {journal} {Journal of Physics A: Mathematical and Theoretical}\ }\textbf {\bibinfo {volume} {51}},\ \bibinfo {pages} {355302} (\bibinfo {year} {2018})}\BibitemShut {NoStop}%
\bibitem [{\citenamefont {Nguyen}\ \emph {et~al.}(2019)\citenamefont {Nguyen}, \citenamefont {Nguyen},\ and\ \citenamefont {G\"uhne}}]{Nguyen2019}%
  \BibitemOpen
  \bibfield  {author} {\bibinfo {author} {\bibfnamefont {H.~C.}\ \bibnamefont {Nguyen}}, \bibinfo {author} {\bibfnamefont {H.-V.}\ \bibnamefont {Nguyen}},\ and\ \bibinfo {author} {\bibfnamefont {O.}~\bibnamefont {G\"uhne}},\ }\bibfield  {title} {\bibinfo {title} {Geometry of einstein-podolsky-rosen correlations},\ }\href {https://doi.org/10.1103/PhysRevLett.122.240401} {\bibfield  {journal} {\bibinfo  {journal} {Phys. Rev. Lett.}\ }\textbf {\bibinfo {volume} {122}},\ \bibinfo {pages} {240401} (\bibinfo {year} {2019})}\BibitemShut {NoStop}%
\bibitem [{\citenamefont {Zhang}\ \emph {et~al.}(2025{\natexlab{b}})\citenamefont {Zhang}, \citenamefont {Zhang},\ and\ \citenamefont {Chitambar}}]{ZhangZhangChitambar2025}%
  \BibitemOpen
  \bibfield  {author} {\bibinfo {author} {\bibfnamefont {Y.}~\bibnamefont {Zhang}}, \bibinfo {author} {\bibfnamefont {J.}~\bibnamefont {Zhang}},\ and\ \bibinfo {author} {\bibfnamefont {E.}~\bibnamefont {Chitambar}},\ }\bibfield  {title} {\bibinfo {title} {Cost of simulating entanglement in steering scenarios},\ }\href@noop {} {\bibfield  {journal} {\bibinfo  {journal} {Quantum}\ }\textbf {\bibinfo {volume} {9}},\ \bibinfo {pages} {1902} (\bibinfo {year} {2025}{\natexlab{b}})},\ \Eprint {https://arxiv.org/abs/2302.09060} {arXiv:2302.09060 [quant-ph]} \BibitemShut {NoStop}%
\bibitem [{\citenamefont {Jevtic}\ \emph {et~al.}(2015)\citenamefont {Jevtic}, \citenamefont {Hall}, \citenamefont {Anderson}, \citenamefont {Zwierz},\ and\ \citenamefont {Wiseman}}]{Jevtic2015}%
  \BibitemOpen
  \bibfield  {author} {\bibinfo {author} {\bibfnamefont {S.}~\bibnamefont {Jevtic}}, \bibinfo {author} {\bibfnamefont {M.~J.~W.}\ \bibnamefont {Hall}}, \bibinfo {author} {\bibfnamefont {M.~R.}\ \bibnamefont {Anderson}}, \bibinfo {author} {\bibfnamefont {M.}~\bibnamefont {Zwierz}},\ and\ \bibinfo {author} {\bibfnamefont {H.~M.}\ \bibnamefont {Wiseman}},\ }\bibfield  {title} {\bibinfo {title} {Einstein--podolsky--rosen steering and the steering ellipsoid},\ }\href {https://doi.org/10.1364/JOSAB.32.000A40} {\bibfield  {journal} {\bibinfo  {journal} {J. Opt. Soc. Am. B}\ }\textbf {\bibinfo {volume} {32}},\ \bibinfo {pages} {A40} (\bibinfo {year} {2015})}\BibitemShut {NoStop}%
\bibitem [{\citenamefont {Nguyen}\ and\ \citenamefont {Vu}(2016)}]{Nguyen2016}%
  \BibitemOpen
  \bibfield  {author} {\bibinfo {author} {\bibfnamefont {H.~C.}\ \bibnamefont {Nguyen}}\ and\ \bibinfo {author} {\bibfnamefont {T.}~\bibnamefont {Vu}},\ }\bibfield  {title} {\bibinfo {title} {Necessary and sufficient condition for steerability of two-qubit states by the geometry of steering outcomes},\ }\href {https://doi.org/10.1209/0295-5075/115/10003} {\bibfield  {journal} {\bibinfo  {journal} {{EPL} (Europhysics Letters)}\ }\textbf {\bibinfo {volume} {115}},\ \bibinfo {pages} {10003} (\bibinfo {year} {2016})}\BibitemShut {NoStop}%
\bibitem [{\citenamefont {Nguyen}\ and\ \citenamefont {G\"uhne}(2020)}]{Nguyen2020}%
  \BibitemOpen
  \bibfield  {author} {\bibinfo {author} {\bibfnamefont {H.~C.}\ \bibnamefont {Nguyen}}\ and\ \bibinfo {author} {\bibfnamefont {O.}~\bibnamefont {G\"uhne}},\ }\bibfield  {title} {\bibinfo {title} {Quantum steering of bell-diagonal states with generalized measurements},\ }\href {https://doi.org/10.1103/PhysRevA.101.042125} {\bibfield  {journal} {\bibinfo  {journal} {Phys. Rev. A}\ }\textbf {\bibinfo {volume} {101}},\ \bibinfo {pages} {042125} (\bibinfo {year} {2020})}\BibitemShut {NoStop}%
\bibitem [{\citenamefont {Chiribella}\ \emph {et~al.}(2007)\citenamefont {Chiribella}, \citenamefont {D'Ariano},\ and\ \citenamefont {Schlingemann}}]{Chiribella2007}%
  \BibitemOpen
  \bibfield  {author} {\bibinfo {author} {\bibfnamefont {G.}~\bibnamefont {Chiribella}}, \bibinfo {author} {\bibfnamefont {G.~M.}\ \bibnamefont {D'Ariano}},\ and\ \bibinfo {author} {\bibfnamefont {D.}~\bibnamefont {Schlingemann}},\ }\bibfield  {title} {\bibinfo {title} {How continuous quantum measurements in finite dimensions are actually discrete},\ }\href {https://doi.org/10.1103/PhysRevLett.98.190403} {\bibfield  {journal} {\bibinfo  {journal} {Phys. Rev. Lett.}\ }\textbf {\bibinfo {volume} {98}},\ \bibinfo {pages} {190403} (\bibinfo {year} {2007})}\BibitemShut {NoStop}%
\bibitem [{\citenamefont {D'Ariano}\ \emph {et~al.}(2005)\citenamefont {D'Ariano}, \citenamefont {Lo~Presti},\ and\ \citenamefont {Perinotti}}]{Ariano2005}%
  \BibitemOpen
  \bibfield  {author} {\bibinfo {author} {\bibfnamefont {G.~M.}\ \bibnamefont {D'Ariano}}, \bibinfo {author} {\bibfnamefont {P.}~\bibnamefont {Lo~Presti}},\ and\ \bibinfo {author} {\bibfnamefont {P.}~\bibnamefont {Perinotti}},\ }\bibfield  {title} {\bibinfo {title} {Classical randomness in quantum measurements},\ }\href {https://doi.org/10.1088/0305-4470/38/26/010} {\bibfield  {journal} {\bibinfo  {journal} {Journal of Physics A: Mathematical and General}\ }\textbf {\bibinfo {volume} {38}},\ \bibinfo {pages} {5979} (\bibinfo {year} {2005})}\BibitemShut {NoStop}%
\bibitem [{\citenamefont {Jevtic}\ \emph {et~al.}(2014)\citenamefont {Jevtic}, \citenamefont {Pusey}, \citenamefont {Jennings},\ and\ \citenamefont {Rudolph}}]{Sania2014}%
  \BibitemOpen
  \bibfield  {author} {\bibinfo {author} {\bibfnamefont {S.}~\bibnamefont {Jevtic}}, \bibinfo {author} {\bibfnamefont {M.}~\bibnamefont {Pusey}}, \bibinfo {author} {\bibfnamefont {D.}~\bibnamefont {Jennings}},\ and\ \bibinfo {author} {\bibfnamefont {T.}~\bibnamefont {Rudolph}},\ }\bibfield  {title} {\bibinfo {title} {Quantum steering ellipsoids},\ }\href {https://doi.org/10.1103/PhysRevLett.113.020402} {\bibfield  {journal} {\bibinfo  {journal} {Phys. Rev. Lett.}\ }\textbf {\bibinfo {volume} {113}},\ \bibinfo {pages} {020402} (\bibinfo {year} {2014})}\BibitemShut {NoStop}%
\bibitem [{\citenamefont {Wiseman}\ \emph {et~al.}(2007)\citenamefont {Wiseman}, \citenamefont {Jones},\ and\ \citenamefont {Doherty}}]{Wiseman2007}%
  \BibitemOpen
  \bibfield  {author} {\bibinfo {author} {\bibfnamefont {H.~M.}\ \bibnamefont {Wiseman}}, \bibinfo {author} {\bibfnamefont {S.~J.}\ \bibnamefont {Jones}},\ and\ \bibinfo {author} {\bibfnamefont {A.~C.}\ \bibnamefont {Doherty}},\ }\bibfield  {title} {\bibinfo {title} {Steering, entanglement, nonlocality, and the einstein-podolsky-rosen paradox},\ }\href {https://doi.org/10.1103/PhysRevLett.98.140402} {\bibfield  {journal} {\bibinfo  {journal} {Phys. Rev. Lett.}\ }\textbf {\bibinfo {volume} {98}},\ \bibinfo {pages} {140402} (\bibinfo {year} {2007})}\BibitemShut {NoStop}%
\bibitem [{\citenamefont {Almeida}\ \emph {et~al.}(2007)\citenamefont {Almeida}, \citenamefont {Pironio}, \citenamefont {Barrett}, \citenamefont {T\'oth},\ and\ \citenamefont {Ac\'{\i}n}}]{Almeida2007}%
  \BibitemOpen
  \bibfield  {author} {\bibinfo {author} {\bibfnamefont {M.~L.}\ \bibnamefont {Almeida}}, \bibinfo {author} {\bibfnamefont {S.}~\bibnamefont {Pironio}}, \bibinfo {author} {\bibfnamefont {J.}~\bibnamefont {Barrett}}, \bibinfo {author} {\bibfnamefont {G.}~\bibnamefont {T\'oth}},\ and\ \bibinfo {author} {\bibfnamefont {A.}~\bibnamefont {Ac\'{\i}n}},\ }\bibfield  {title} {\bibinfo {title} {Noise robustness of the nonlocality of entangled quantum states},\ }\href {https://doi.org/10.1103/PhysRevLett.99.040403} {\bibfield  {journal} {\bibinfo  {journal} {Phys. Rev. Lett.}\ }\textbf {\bibinfo {volume} {99}},\ \bibinfo {pages} {040403} (\bibinfo {year} {2007})}\BibitemShut {NoStop}%
\bibitem [{\citenamefont {Ioannou}\ \emph {et~al.}(2022)\citenamefont {Ioannou}, \citenamefont {Sekatski}, \citenamefont {Designolle}, \citenamefont {Jones}, \citenamefont {Uola},\ and\ \citenamefont {Brunner}}]{Ioannou2022}%
  \BibitemOpen
  \bibfield  {author} {\bibinfo {author} {\bibfnamefont {M.}~\bibnamefont {Ioannou}}, \bibinfo {author} {\bibfnamefont {P.}~\bibnamefont {Sekatski}}, \bibinfo {author} {\bibfnamefont {S.}~\bibnamefont {Designolle}}, \bibinfo {author} {\bibfnamefont {B.~D.~M.}\ \bibnamefont {Jones}}, \bibinfo {author} {\bibfnamefont {R.}~\bibnamefont {Uola}},\ and\ \bibinfo {author} {\bibfnamefont {N.}~\bibnamefont {Brunner}},\ }\bibfield  {title} {\bibinfo {title} {Simulability of high-dimensional quantum measurements},\ }\href {https://doi.org/10.1103/PhysRevLett.129.190401} {\bibfield  {journal} {\bibinfo  {journal} {Phys. Rev. Lett.}\ }\textbf {\bibinfo {volume} {129}},\ \bibinfo {pages} {190401} (\bibinfo {year} {2022})}\BibitemShut {NoStop}%
\bibitem [{\citenamefont {Fujiwara}\ and\ \citenamefont {Algoet}(1999)}]{Fujiwara1999}%
  \BibitemOpen
  \bibfield  {author} {\bibinfo {author} {\bibfnamefont {A.}~\bibnamefont {Fujiwara}}\ and\ \bibinfo {author} {\bibfnamefont {P.}~\bibnamefont {Algoet}},\ }\bibfield  {title} {\bibinfo {title} {One-to-one parametrization of quantum channels},\ }\href {https://doi.org/10.1103/PhysRevA.59.3290} {\bibfield  {journal} {\bibinfo  {journal} {Phys. Rev. A}\ }\textbf {\bibinfo {volume} {59}},\ \bibinfo {pages} {3290} (\bibinfo {year} {1999})}\BibitemShut {NoStop}%
\bibitem [{\citenamefont {Bowles}\ \emph {et~al.}(2016)\citenamefont {Bowles}, \citenamefont {Hirsch}, \citenamefont {Quintino},\ and\ \citenamefont {Brunner}}]{Bowles2016}%
  \BibitemOpen
  \bibfield  {author} {\bibinfo {author} {\bibfnamefont {J.}~\bibnamefont {Bowles}}, \bibinfo {author} {\bibfnamefont {F.}~\bibnamefont {Hirsch}}, \bibinfo {author} {\bibfnamefont {M.~T.}\ \bibnamefont {Quintino}},\ and\ \bibinfo {author} {\bibfnamefont {N.}~\bibnamefont {Brunner}},\ }\bibfield  {title} {\bibinfo {title} {Sufficient criterion for guaranteeing that a two-qubit state is unsteerable},\ }\href {https://doi.org/10.1103/PhysRevA.93.022121} {\bibfield  {journal} {\bibinfo  {journal} {Phys. Rev. A}\ }\textbf {\bibinfo {volume} {93}},\ \bibinfo {pages} {022121} (\bibinfo {year} {2016})}\BibitemShut {NoStop}%
\bibitem [{\citenamefont {Villegas-Aguilar}\ \emph {et~al.}(2024)\citenamefont {Villegas-Aguilar}, \citenamefont {Polino}, \citenamefont {Ghafari}, \citenamefont {Quintino}, \citenamefont {Laverick}, \citenamefont {Berkman}, \citenamefont {Rogge}, \citenamefont {Shalm}, \citenamefont {Tischler}, \citenamefont {Cavalcanti}, \citenamefont {Slussarenko},\ and\ \citenamefont {Pryde}}]{Aguilar2024}%
  \BibitemOpen
  \bibfield  {author} {\bibinfo {author} {\bibfnamefont {L.}~\bibnamefont {Villegas-Aguilar}}, \bibinfo {author} {\bibfnamefont {E.}~\bibnamefont {Polino}}, \bibinfo {author} {\bibfnamefont {F.}~\bibnamefont {Ghafari}}, \bibinfo {author} {\bibfnamefont {M.~T.}\ \bibnamefont {Quintino}}, \bibinfo {author} {\bibfnamefont {K.~T.}\ \bibnamefont {Laverick}}, \bibinfo {author} {\bibfnamefont {I.~R.}\ \bibnamefont {Berkman}}, \bibinfo {author} {\bibfnamefont {S.}~\bibnamefont {Rogge}}, \bibinfo {author} {\bibfnamefont {L.~K.}\ \bibnamefont {Shalm}}, \bibinfo {author} {\bibfnamefont {N.}~\bibnamefont {Tischler}}, \bibinfo {author} {\bibfnamefont {E.~G.}\ \bibnamefont {Cavalcanti}}, \bibinfo {author} {\bibfnamefont {S.}~\bibnamefont {Slussarenko}},\ and\ \bibinfo {author} {\bibfnamefont {G.~J.}\ \bibnamefont {Pryde}},\ }\bibfield  {title} {\bibinfo {title} {Nonlocality activation in a photonic quantum network},\ }\href {https://doi.org/10.1038/s41467-024-47354-w} {\bibfield  {journal} {\bibinfo  {journal} {Nature
  Communications}\ }\textbf {\bibinfo {volume} {15}},\ \bibinfo {pages} {3112} (\bibinfo {year} {2024})}\BibitemShut {NoStop}%
\end{thebibliography}%

\appendix
\section{Proof of Proposition~\ref{prop:orthogonality-normal-form}}
\label{app:orthogonality-normal-form}

For a finite Hermitian family $\{Y_a\}_{a=1}^m$, consider
\begin{align}
\sup_{\substack{M_a\ge0\\ \sum_aM_a=\mbb{1}}}
\sum_a\tr \left[\mc N(Y_a)M_a\right].
\end{align}
Introducing a Hermitian multiplier $Z$ for the constraint
$\sum_aM_a=\mbb{1}$ gives the Lagrangian
\begin{align}
L(\{M_a\},Z)&=
\sum_a\tr \left[\mc N(Y_a)M_a\right]+\tr \left[Z\left(\mbb{1}-\sum_aM_a\right)\right]\notag\\
&=\tr Z+\sum_a\tr \left[\bigl(\mc N(Y_a)-Z\bigr)M_a\right].
\end{align}
The supremum over $M_a\ge0$ is finite if and only if
\begin{align}
Z\ge\mc N(Y_a)\qquad\forall a,
\end{align}
in which case it equals $\tr Z$. Hence the dual problem is
\begin{align}
\inf_Z\left\{\tr Z:Z\ge\mc N(Y_a)\ \forall a\right\}.
\end{align}
Since $M_a=\mbb{1}/m$ is strictly feasible, Slater's condition gives
strong duality, proving Eq.~\eqref{eq:pri-dual}.

Now assume first that Eq.~\eqref{eq:orthogonality-normal-form} holds. Fix a
finite Hermitian family $\{Y_a\}_a$, and let $Z_*$ and POVM $\{Q_a\}_a$ be
optimal solutions of the dual and primal SDPs in
Eq.~\eqref{eq:pri-dual}, respectively. Define
\begin{align}
P_a:=Z_*-\mc N(Y_a).
\label{eq:general-Pa-definition}
\end{align}
Dual feasibility gives $P_a\ge\mbb{0}$. Since the primal SDP is
strictly feasible, Slater's condition ensures strong duality and the
KKT conditions. Complementary slackness gives
$\tr(P_aQ_a)=0$. For positive operators $P_a,Q_a$, this implies
\begin{align}
P_aQ_a=\mbb{0}
\qquad
\forall a.
\label{eq:general-complementary-slackness}
\end{align}

Since $\mc N$ is invertible, we can write
\begin{align}
Y_a=\mc N^{-1}(Z_*)-\mc N^{-1}(P_a).
\end{align}
The first term is independent of $a$, and hence
\begin{align}
&\int_\Lambda
\max_a\tr(Y_a\Pi_\lambda)d\nu(\lambda)\label{eq:general-exact-reduction}\\
&=\tr \left[\mc N^{-1}(Z_*)\right]+
\int_\Lambda\max_a\tr \left[-\mc N^{-1}(P_a)\Pi_\lambda\right]d\nu(\lambda).
\notag 
\end{align}
Here we used $\int_\Lambda \Pi_\lambda d\nu(\lambda)=\mbb{1}$.
Moreover, since $\mc N$ is trace preserving, its inverse is also trace
preserving, so
\begin{align}
\tr \left[\mc N^{-1}(Z_*)\right]=\tr Z_*.
\end{align}
Equation~\eqref{eq:orthogonality-normal-form} therefore implies
\begin{align}
\int_\Lambda
\max_a\tr(Y_a\Pi_\lambda)d\nu(\lambda)
\ge
\tr Z_*,
\end{align}
which is precisely Eq.~\eqref{eq:dual-charac}.

Conversely, suppose that Eq.~\eqref{eq:dual-charac} holds for every
finite Hermitian family. Let $P_a\ge\mbb{0}$ be any finite family for
which there exists a POVM $\{Q_a\}_a$ satisfying
$P_aQ_a=\mbb{0}$ for every $a$, and define
\begin{align}
Y_a:=-\mc N^{-1}(P_a).
\end{align}
Then $\mc N(Y_a)=-P_a$, so $Z=\mbb{0}$ is feasible in the dual SDP. Furthermore, every feasible $Z$ satisfies
\begin{align}
\tr Z=\sum_a\tr(ZQ_a)\ge-\sum_a\tr(P_aQ_a)=0.
\end{align}
Thus $Z=\mbb{0}$ is optimal. Applying
Eq.~\eqref{eq:dual-charac} to this witness family gives
\begin{align}
\int_\Lambda
\max_a
\tr \left[
-\mc N^{-1}(P_a)\Pi_\lambda
\right]
d\nu(\lambda)
\ge0,
\end{align}
which proves Eq.~\eqref{eq:orthogonality-normal-form}.

\section{Proofs for the main results}
\label{app:main-results}
\subsection{Singular channels}
\subsubsection{Singular qubit unital channels at the boundary}
\label{app:singular-boundary}

We next justify the direct treatment of singular channels at the boundary. Up to input and output unitary rotations, $D$ has diagonal form
$\text{diag}(\eta_1,\eta_2,0)$. The Fujiwara--Algoet complete-positivity
conditions~\cite{Fujiwara1999} give
\begin{align}
|\eta_1+\eta_2|\le1,\quad|\eta_1-\eta_2|\le1,
\end{align}
and hence
\begin{align}
|\eta_1|+|\eta_2|=\max\left\{|\eta_1+\eta_2|,
|\eta_1-\eta_2|\right\}\le1.
\label{eq:singular-cp-bound}
\end{align}
Rotational invariance then gives
\begin{align}
&2\int_{S^2}\norm{D\hat n}d\mu(\hat n)=2\int_{S^2}\sqrt{\eta_1^2n_1^2+\eta_2^2n_2^2} d\mu(\hat n)
\notag\\
&\le2\int_{S^2}\left(|\eta_1n_1|+|\eta_2n_2|\right)d\mu(\hat n)=|\eta_1|+|\eta_2|\le1.
\label{eq:singular-average-bound}
\end{align}
Since $\sqrt{\eta_1^2n_1^2+\eta_2^2n_2^2}<|\eta_1n_1|+|\eta_2n_2|$ whenever $\eta_1\eta_2n_1n_2\ne0$,  the equality at the boundary holds only if $|\eta_1|=1$ and $\eta_2=0$. Therefore $D$ has singular values $(1,0,0)$. There must exist $\hat u,\hat v\in S^2$ such that
\begin{align}
D=\hat u\hat v^T.
\end{align}
For an arbitrary qubit effect
$M_a=x_{0,a}\mbb{1}+\vec x_a\cdot\vec\sigma$, one has
\begin{align}
\mc N_D^\dagger(M_a)&=x_{0,a}\mbb{1}+(\hat v\cdot\vec x_a)\hat u\cdot\vec\sigma,\\
&=[x_{0,a}+(\hat v\cdot\vec x_a)]P_++[x_{0,a}-(\hat v\cdot\vec x_a)]P_-\notag
\end{align}
which can be simulated by the parent $\{P_{\pm}=\frac{1}{2}(\mbb1 \pm \hat{u}\cdot\vec \sigma)\}$ since $x_{0,a}\pm (\hat v\cdot\vec x_a)$ are nonnegative. Therefore, $\mc N_D$ maps every set of measurements to a jointly measurable set, and is incompatibility-breaking.  
\subsubsection{Arbitrary Singular Qubit Channels and Closure}
\label{app:singular-general}
More generally, although Proposition~\ref{prop:orthogonality-normal-form} requires $\mc N$ to be invertible as a linear map, this assumption entails no loss for the incompatibility-breaking criteria proved below. To see
this, let $\mc N_{D,\vec t}$ be a qubit channel satisfying
\begin{align}
2\int_{S^2}\norm{D\hat n}d\mu(\hat n)+\norm{\vec t}\le1,
\end{align}
where the unital case is recovered by setting $\vec t=\vec0$. Let $\Delta_\eta$ be the partial depolarizing channel with Bloch matrix $\eta I$, where $0<\eta<1/2$, and define
\begin{align}
\mc N_\epsilon:=(1-\epsilon)\mc N_{D,\vec t}+\epsilon\Delta_\eta .
\end{align}
Since $\mc N_\epsilon$ is a convex combination of channels, it
is again a valid channel. The Pauli representations of $\mc N_\epsilon$ are
\begin{align}
D_\epsilon=(1-\epsilon)D+\epsilon\eta I,\qquad\vec t_\epsilon=(1-\epsilon)\vec t,
\end{align}
and they satisfy
\begin{align}
2\int_{S^2}\norm{D_\epsilon\hat n}d\mu(\hat n)+\norm{\vec t_\epsilon}\le1-\epsilon(1-2\eta)<1.
\end{align}
Since $\det[D_\epsilon]$ above is a polynomial in $\epsilon$ and is not identically zero, there exists a sequence $\epsilon_k\rightarrow0$ such that every
$D_{\epsilon_k}$ is invertible.

Let $\Pi_k$ be the parent in Eq.~\eqref{eq:nonunital-parent-main} constructed from
$(D_{\epsilon_k},\vec t_{\epsilon_k})$. Then
$\Pi_k(\hat n)\to \Pi_{D,\vec t}(\hat n)$ uniformly on $S^2$.

For every finite Hermitian family $\{Y_a\}_a$, the invertible case
gives
\begin{align}
&\sup_{\{M_a\}}\sum_a\tr\left[\mc N_{\epsilon_k}(Y_a)M_a\right]\notag\le\int_{S^2}\max_a\tr\left[Y_a\Pi_k(\hat n)\right]d\mu(\hat n).
\end{align}
The left-hand side converges because the POVM set is compact and the
objective depends continuously on the channel. The right-hand side
converges by the uniform convergence of $\Pi_k$. Since this holds for every finite Hermitian family $\{Y_a\}_a$, Eq.~\eqref{eq:dual-charac} implies that the explicit parent $\Pi_{D,\vec t}$ simulates every noisy POVM under $\mc N_{D,\vec t}^\dagger$.



This proves Theorem~\ref{thm:main} and Proposition~\ref{prop:nonunital-sufficient} for singular $D$ and, moreover, shows that the same explicit parent is retained in the limit.

\subsection{Steering--joint-measurability correspondence and steering criteria}
\label{app:state-channel-correspondence}

Consider a general two-qubit state
\begin{align}
\rho_{AB}=\frac{1}{4}\left[\mbb1\otimes\mbb1+\vec a\cdot\vec\sigma\otimes\mbb1+\mbb1\otimes\vec b\cdot\vec\sigma+\sum_{j,k=1}^3T_{jk}\sigma_j\otimes\sigma_k\right].
\end{align}
For steering from Alice to Bob, if $\rho_B$ is full rank, the
invertible local filter
\begin{align}
\widetilde\rho_{AB}=\left[\mbb1\otimes(2\rho_B)^{-1/2}\right]\rho_{AB}\left[\mbb1\otimes(2\rho_B)^{-1/2}\right]
\end{align}
preserves steerability and gives
$\widetilde\rho_B=\mbb1/2$~\cite{Sania2014}. Dropping tildes, the
filtered state has the canonical form in
Eq.~\eqref{eq:canonical-two-qubit-state}. If $\rho_B$ has rank one,
then $\rho_{AB}$ is a product state and is automatically unsteerable.

Let $\Theta:=\operatorname{diag}(1,-1,1)$. The steering–joint-measurability correspondence~\cite{Uola2015,Kiukas2017} associates the canonical
state with the qubit channel $\mc N_{D,\vec t}$ defined by
\begin{align}
\vec t=\Theta\vec a,
\qquad
D=T^T\Theta .
\end{align}
Indeed, for every qubit effect
$M=x_0\mbb1+\vec x\cdot\vec\sigma$, direct calculation gives
\begin{align}
\tr_A\left[(M\otimes\mbb1)\rho_{AB}\right]=\frac{1}{2}\mc N_{D,\vec t}^\dagger(M^T).
\label{eq:state-channel-correspondence}
\end{align}
Since transposition is a bijection on the set of POVMs, an LHS model
for the resulting assemblage is equivalent to a parent POVM for
$\{\mc N_{D,\vec t}^\dagger(M_{a|x}^T)\}_{a,x}$: one simply identifies the effects of the parent POVM as $\Pi_\lambda=2\rho_\lambda$, normalization then follows from $\int_\Lambda\rho_\lambda d\nu(\lambda)=\mbb1/2$.

Consequently, the canonical state is unsteerable under arbitrary POVMs if and only if $\mc N_{D,\vec t}$ is
incompatibility breaking. Since $\Theta$ is orthogonal and $D$ and
$T$ have the same singular values,
Proposition~\ref{prop:nonunital-sufficient} gives
Eq.~\eqref{eq:general-two-qubit-unsteerability}. When $\vec a=\vec0$,
the associated channel is unital, so Theorem~\ref{thm:main} makes the
condition necessary and sufficient. This proves
Theorem~\ref{thm:two-qubit-steering}.

\subsubsection{Nontrivial unsteerable two-qubit state}
\label{app:nontrivial-nonunital-example}

Let $\mc U_0$ denote the set of physical canonical states with
$\vec a=\vec0$ satisfying
\begin{align}
m(T):=2\int_{S^2}\norm{T\hat n}d\mu(\hat n)\le1,
\end{align}
and let $\mathrm{Sep}$ denote the set of separable two-qubit states.
Consider the Hermitian witness
\begin{align}
H=&-\mbb1\otimes\mbb1+\frac{39}{200}\sigma_z\otimes\mbb1-\frac9{50}\mbb1\otimes\sigma_z\notag\\&+\frac{49}{50}
\left(\sigma_x\otimes\sigma_x+\sigma_y\otimes\sigma_y\right)+\frac{23}{25}\sigma_z\otimes\sigma_z .
\label{eq:nonunital-separating-witness}
\end{align}
For the state in the main text,
\begin{align}
\tr\left(H\rho^{AB}\right)=\frac1{2000}>0.
\label{eq:nonunital-target-separation}
\end{align}

We first consider $\tau\in\mc U_0$, with correlation matrix $T$.
Positivity of $\tau$, evaluated on the singlet state, gives
\begin{align}
T_{xx}+T_{yy}+T_{zz}\le1.
\end{align}
Moreover,
\begin{align}
\frac{2T_{xx}+2T_{yy}-T_{zz}}{3}
\le m(T)\le1,
\end{align}
where the inequality follows from $\norm{T\hat n}\ge \hat{w}(\hat n)\cdot T\hat n$ with $\hat{w}(\hat n):=\left(\frac{2}{3}\text{sign}(n_x),\frac{2}{3}\text{sign}(n_y),-\frac{1}{3}\text{sign}(n_z)\right)$.

Consequently,
\begin{align}
\tr(H\tau)=&-1+\frac3{50}\left(\frac{2T_{xx}+2T_{yy}-T_{zz}}{3}\right)\notag\\
&+\frac{47}{50}\left(T_{xx}+T_{yy}+T_{zz}\right)\le0.
\end{align}
For the separable set, convexity reduces the optimization to pure product states
\begin{align}
\rho_{\hat u}\otimes\rho_{\hat v}
=
\frac12\left(\mbb1+\hat u\cdot\vec\sigma\right)
\otimes
\frac12\left(\mbb1+\hat v\cdot\vec\sigma\right).
\end{align}
For fixed $\hat u$, direct calculation gives
\begin{align}
\tr\left(H\rho_{\hat u}\otimes\rho_{\hat v}\right)
={}&-1+\frac{39}{200}u_z
+\frac{49}{50}\left(u_xv_x+u_yv_y\right)\notag\\
&+\left(-\frac9{50}+\frac{23}{25}u_z\right)v_z.
\label{eq:product-witness-expectation}
\end{align}
The terms depending on $\hat v$ are maximized when $\hat v\propto\left(\frac{49}{50}u_x,\frac{49}{50}u_y,-\frac9{50}+\frac{23}{25}u_z
\right)$. Now writing $s:=u_z\in[-1,1]$ and using
$u_x^2+u_y^2=1-s^2$, we obtain
\begin{align}
&\tr\left(H\rho_{\hat u}\otimes\rho_{\hat v}\right)
\le-1+\frac{39}{200}s\notag\\
&+\sqrt{\left(\frac{49}{50}\right)^2(1-s^2)+\left(-\frac9{50}+\frac{23}{25}s\right)^2}.
\label{eq:product-witness-one-variable}
\end{align}
A direct calculation shows that the right-hand side is strictly negative for all $s\in[-1,1]$, hence,
\begin{align}
\tr(H\sigma)\le0\qquad\forall\sigma\in\conv\left(\mc U_0\cup\mathrm{Sep}\right),
\end{align}
whereas Eq.~\eqref{eq:nonunital-target-separation} is positive. Hence
\begin{align}
\rho^{AB}\notin\conv\left(\mc U_0\cup\mathrm{Sep}\right).
\end{align}

It is worth noting another complementary way to enlarge known unsteerable regions through positive maps on the trusted party. If $\tau^{AB}$ is unsteerable from Alice to Bob and $\Gamma$ is a positive linear map such that
\begin{align}
\rho^{AB}=\frac{(\mbb I\otimes\Gamma)(\tau^{AB})}
{\tr[(\mbb I\otimes\Gamma)(\tau^{AB})]}
\end{align}
is a valid quantum state, then $\rho^{AB}$ is also unsteerable from Alice to Bob~\cite{Bowles2016}; a related positive-map construction was used in Ref.~\cite{Aguilar2024}. In fact, the state considered above admits an alternative certification of this form, with $\tau^{AB}\in\conv(\mc U_0\cup\mathrm{Sep})$ and a suitable positive map $\Gamma$. Thus, although $\rho^{AB}\notin\conv(\mc U_0\cup\mathrm{Sep})$, it belongs to the positive-map closure of this set.

\subsection{Proof of Lemma~\ref{lem:weighted-ellipsoidal-average}}
\label{app:weighted-ellipsoidal-average}

\begin{proof}
For every nonempty $J\subseteq\mc I$, define
\begin{align}
F_J:=\int_{S^2}\max_{i\in J}\alpha_i\left(\hat u_i\cdot\hat n-h(\hat n)\right)d\mu(\hat n).
\end{align}
Since $J\subseteq\mc I$, we have $F_{\mc I}\ge F_J$. By Carath\'eodory's theorem, there exists $J\subseteq\mc I$, with $|J|\le4$, such that
\begin{align}
\vec0\in\conv\{\hat u_i:i\in J\}.
\label{eq:reduced-balanced-subset}
\end{align}
Since every $\hat u_i$ is a unit vector, necessarily $|J|\ge2$. It is therefore sufficient to consider the cases $|J|=2,3,4$.

Suppose first that $|J|=2$. Since $ \vec0\in\conv\{\hat u_i:i\in J\}$, we can write them as $\hat u$ and $-\hat u$, and define
\begin{align}
\delta:=\frac{\alpha_2-\alpha_1}{\alpha_1+\alpha_2}.
\end{align}
Using $\max\{A,B\}={(A+B+|A-B|)}/{2}$, a direct calculation gives
\begin{align}
F_{\{1,2\}}&=\frac{\alpha_1+\alpha_2}{2}\left[\int_{S^2}\left|\hat u\cdot\hat n+\delta h(\hat n)\right|
d\mu(\hat n)-\frac{1}{2}\right]\notag \\
&\ge \frac{\alpha_1+\alpha_2}{2}\left[\int_{S^2}\left|\hat u\cdot\hat n\right|
d\mu(\hat n)-\frac{1}{2}\right]\ge 0,
\end{align}
where the first inequality follows from the fact that both $h(\hat n)$, $d\mu(\hat{n})$ are even (i.e., invariant under $\hat n\mapsto-\hat n$) and $|x+a|+|x-a|\ge 2|x|$. 

Suppose next that $|J|=3$. Relabel the indices so that
\begin{align}
\alpha_1\le \alpha_2\le \alpha_3,
\end{align}
Then there exists $t\in[0,1]$ such that
\begin{align}
\alpha_2=(1-t)\alpha_1+t\alpha_3.
\label{eq:middle-weight}
\end{align}
Since both $h(\hat n)$ and $d\mu(\hat{n})$ are even, 
\begin{align}
F_J=&\int_{S^2}\max_{i\in J}\alpha_i\left(\hat u_i\cdot\hat n-h(\hat n)\right)d\mu(\hat n)\notag\\ 
\overset{\hat n\rightarrow -\hat n}{=}&\int_{S^2}\max_{i\in J}\alpha_i\left(-\hat u_i\cdot\hat n-h(\hat n)\right)d\mu(\hat n)
\end{align}
Averaging the two representations of $F_{J}$ gives
\begin{align}
&2F_{J}=\int_{S^2}\max_{i,j\in J}A_{ij}(\hat n)d\mu(\hat n), \label{eq:three-antipodal} \\
\text{where }&A_{ij}(\hat n):=(\alpha_i\hat u_i-\alpha_j\hat u_j)\cdot\hat n-(\alpha_i+\alpha_j)h(\hat n).\notag
\end{align}
The symmetrization pairs $A_{ij}$ with $A_{ji}$, whose linear parts have opposite signs while their coefficients of $h(\hat n)$ are identical. In the following, we use two pointwise lower bounds of Eq.~\eqref{eq:three-antipodal}. 

For the first, a maximum dominates every convex combination of its entries, so pointwise
\begin{align}
\max_{i,j}A_{ij}(\hat n)
\ge&\max\Bigl\{A_{31}(\hat n),tA_{12}(\hat n)+(1-t)A_{32}(\hat n),\notag \\
&A_{13}(\hat n),tA_{21}(\hat n)+(1-t)A_{23}(\hat n)
\Bigr\}.
\label{eq:selected-pair-branches}
\end{align}
Define
\begin{subequations}
\begin{align}
\vec v&:=\alpha_3\hat u_3-\alpha_1\hat u_1,\\
\vec w&:=t(\alpha_1\hat u_1-\alpha_2\hat u_2)+(1-t)(\alpha_3\hat u_3-\alpha_2\hat u_2).
\end{align}
\label{eq:def vw}
\end{subequations}
The four expressions in
Eq.~\eqref{eq:selected-pair-branches} reduce to
\begin{align}
\pm \vec v\cdot\hat n-(\alpha_1+\alpha_3)h(\hat n),~\pm \vec w\cdot\hat n-(\alpha_1+\alpha_3)h(\hat n).
\end{align}
The common $h(\hat n)$-term can therefore be separated from the maximization.

Using $\max\{|a|,|b|\}=\frac{|a+b|+|a-b|}{2}$ and Eq.~\eqref{eq:selected-pair-branches},
\begin{align}
&\max_{i,j}A_{ij}(\hat n)\ge\max\left\{|\vec v\cdot\hat n|,|\vec w\cdot\hat n|\right\}-(\alpha_1+\alpha_3)h(\hat n)\notag \\
&=\frac{|\vec v\cdot\hat n+\vec w\cdot\hat n|+|\vec v\cdot\hat n-\vec w\cdot\hat n|}{2}-(\alpha_1+\alpha_3)h(\hat n)
\label{eq:common-penalty-envelope}
\end{align}
Using $\int_{S^2}|\vec z\cdot\hat n|d\mu(\hat n)=\frac{\norm{\vec z}}{2}$, Eqs.~\eqref{eq:three-antipodal},
\eqref{eq:common-penalty-envelope} and the assumption that $\int_{S^2}h({\hat n})d\mu(\hat n)=\frac{1}{2}$ give
\begin{align}
2F_{\{1,2,3\}}\ge\frac{\norm{\vec v-\vec w}+\norm{\vec v+\vec w}}{4}-\frac{\alpha_1+\alpha_3}{2}.
\label{eq:three-width-bound}
\end{align}

For the second relaxation, we bound
\begin{align}
&\norm{\vec v-\vec w}+\norm{\vec v+\vec w}=\max_{\substack{\norm{x}\le1\\\norm{y}\le1}}
\left[x\cdot(\vec v-\vec w)+y\cdot(\vec v+\vec w)\right].\notag \\
&\ge-\hat u_1\cdot(\vec v-\vec w)+\hat u_3\cdot(\vec v+\vec w)\notag\\
&=2(\alpha_1+\alpha_3)-\alpha_2\left(1+\hat u_1\cdot\hat u_2+\hat u_2\cdot\hat u_3+\hat u_3\cdot\hat u_1\right)\notag\\
&=2(\alpha_1+\alpha_3)+\frac{\alpha_2}{2}\left[1-|\hat u_1+\hat u_2+\hat u_3|^{2}\right].
\end{align}
where the second equality uses Eqs.~\eqref{eq:def vw} and \eqref{eq:middle-weight}.

Eq.~\eqref{eq:three-width-bound} yields
\begin{align}
F_{\{1,2,3\}}\ge\frac{\alpha_2}{16}\left[1-\norm{\hat u_1+\hat u_2+\hat u_3}^{2}
\right].
\label{eq:three-quantitative}
\end{align}
Since $\vec0\in\conv\{\hat u_1,\hat u_2,\hat u_3\}$, one can find $\beta_i\ge 0$ such that $\sum_{i=1}^3\beta_i=1$ and $\sum_{i=1}^3\beta_i\hat u_i=\vec0$. 

The triangle inequality gives that $\beta_i\le \sum_{j\ne i}\beta_j$, therefore, $\beta_i\le 1/2$ for all $i$, and one can then define
\begin{align}
q_i:=1-2\beta_i\ge0,
\qquad
\sum_{i=1}^3q_i=1.
\end{align}
such that $\sum_iq_i\hat u_i=\sum_i\hat u_i-2\sum_i\beta_i\hat u_i=\sum_i\hat u_i$, hence
\begin{align}
\norm{\hat u_1+\hat u_2+\hat u_3}=\norm{\sum_{i=1}^3q_i\hat u_i}\le 1.
\end{align}
Eq.~\eqref{eq:three-quantitative} therefore gives
\begin{align}
F_{\{1,2,3\}}\ge0.
\end{align}

Eq.~\eqref{eq:three-quantitative} was derived for arbitrary triples of unit vectors; the convex-hull condition was used only afterward to establish the norm bound.

Finally, suppose that $|J|=4$. Since $\vec0\in\conv\{\hat u_1,\hat u_2,\hat u_3,\hat u_4\}$, choose $\beta_i\ge0$ such that $\sum_{i=1}^4\beta_i=1$ and $\sum_{i=1}^4\beta_i\hat u_i=\vec0$.

The fact that $\norm{\beta_i\hat{u}_i}\le \sum_{j\ne i}\norm{\beta_j\hat u_j}$ implies that $\beta_i\le \sum_{j\ne i}\beta_j$, therefore, $\beta_i\le 1/2$ for all $i$, and define
\begin{align}
q_i:=1-2\beta_i\ge0,
\qquad
\sum_{i=1}^4q_i=2.
\end{align}
Now set
\begin{align}
\vec s:=\sum_{i=1}^4\hat u_i=\sum_{i=1}^4q_i\hat u_i.
\end{align}
It follows that
\begin{align}
\sum_{i=1}^4q_i\norm{\vec s-\hat u_i}^{2}&=\left(\sum_{i=1}^4q_i\right)\left(\norm{\vec s}^{2}+1\right)-2\vec s\cdot\sum_{i=1}^4q_i\hat u_i\notag \\
&=2(\norm{\vec s}^2+1)-2\norm{\vec s}^2=2.
\end{align}
Since the weights $\sum_iq_i/2=1$, there must exist $k$ such that
\begin{align}
\norm{\vec s-\hat u_k}=\norm{\sum_{i\ne k}\hat u_i}\le1.
\end{align}
After relabeling the three indices in $J\setminus\{k\}$ according to their associated coefficients $\alpha_1\le \alpha_2\le \alpha_3$, Eq.~\eqref{eq:three-quantitative} applies
and gives
\begin{align}
F_{J\setminus\{k\}}\ge0.
\end{align}
Since $F_J\ge F_{J\setminus\{k\}}$, it follows that $F_J\ge0$ for all $|J|\le 4$.
Together with the Carath\'eodory reduction, this proves the lemma.

\end{proof}
\subsection{Proof of Proposition~\ref{prop:nonunital-sufficient}}
\label{app:nonunital-sufficient}

\begin{proof}
Define
\begin{subequations}
\begin{align}
m_D:&=\int_{S^2}\norm{D\hat n} d\mu(\hat n).\\
\delta_{D,\vec t}:&=\frac{1}{2}-m_D-\frac{\norm{\vec t}}{2}
\end{align}
\end{subequations}
Under the assumption in Eq.~\eqref{eq:nonunital-sufficient-condition}, $\delta_{D,\vec t}\ge 0$. We can then introduce the function
\begin{align}
h_{D,\vec t}(\hat n):=\norm{D\hat n}+\left|\vec t\cdot\hat n\right|+\delta_{D,\vec t}.
\label{eq:nonunital-even-function}
\end{align}
This function is nonnegative and even. Moreover, using $\int_{S^2}
\left|\vec t\cdot\hat n\right|d\mu(\hat n)=\frac{\norm{\vec t}}{2}$, we obtain
\begin{align}
\int_{S^2}h_{D,\vec t}(\hat n)d\mu(\hat n)=m_D+\frac{\norm{\vec t}}{2}+\delta_{D,\vec t}=\frac{1}{2}.
\label{eq:nonunital-h-normalization}
\end{align}

Consider the operator density
\begin{align}
\Pi_{D,\vec t}(\hat n)d\mu(\hat n):=2\left[\left(h_{D,\vec t}(\hat n)+\vec t\cdot\hat n\right)\mbb{1}+(D\hat n)\cdot\vec\sigma\right]d\mu(\hat n),
\label{eq:nonunital-parent}
\end{align}
we can verify that Eq.~\eqref{eq:nonunital-parent} defines a valid
POVM. Indeed, positivity of $\Pi_{D,\vec t}(\hat n)$ follows from
\begin{align}
&h_{D,\vec t}(\hat n)+\vec t\cdot\hat n=\norm{D\hat n}+\left|\vec t\cdot\hat n\right|+\vec t\cdot\hat n+\delta_{D,\vec t}\ge\norm{D\hat n} \notag \\
&\Rightarrow \left(h_{D,\vec t}(\hat n)+\vec t\cdot\hat n\right)\mbb{1}+(D\hat n)\cdot\vec\sigma\ge\mbb{0},
\label{eq:nonunital-parent-positivity}
\end{align}
 for every $\hat n\in S^2$, since a qubit operator $c_0\mbb{1}+\vec c\cdot\vec\sigma$ is positive if and only if $c_0\ge\norm{\vec c}$.

Normalization follows from
Eq.~\eqref{eq:nonunital-h-normalization}:
\begin{align}
\int_{S^2}\Pi_{D,\vec t}(\hat n)d\mu(\hat n)=2\left[\int_{S^2}h_{D,\vec t}(\hat n)d\mu(\hat n)\right]\mbb{1}=\mbb{1}.
\end{align}

It remains to verify
Proposition~\ref{prop:orthogonality-normal-form}. Assume first that
$D$ is invertible. As in the proof of Theorem~\ref{thm:main}, it is
enough to consider the active indices $\mc A=\{a:Q_a\ne\mbb{0}\}$.
The case $P_a=\mbb{0}$ for some $a$ is immediate, while otherwise
positivity and $P_aQ_a=\mbb{0}$ imply, for all $a\in \mc A$
\begin{align}
P_a&=\alpha_a\left(\mbb{1}-\hat u_a\cdot\vec\sigma\right),
\qquad
Q_a=\beta_a\left(\mbb{1}+\hat u_a\cdot\vec\sigma\right),
\label{eq:nonunital-PQ-form}
\end{align}
for some $\alpha_a,\beta_a>0$ and $\hat u_a\in S^2$. Moreover,
$\sum_{a\in\mc A}Q_a=\mbb{1}$ implies
\begin{align}
\vec0\in\conv\{\hat u_a:a\in\mc A\}.
\label{eq:nonunital-balance}
\end{align}

The inverse map is given by $\mc N_{D,\vec t}^{-1}\left(y_0\mbb{1}+\vec y\cdot\vec\sigma\right)=y_0\mbb{1}+D^{-T}\left(\vec y-y_0\vec t\right)\cdot\vec\sigma$, and therefore
\begin{align}
\mc N_{D,\vec t}^{-1}(P_a)=\alpha_a\left[\mbb{1}-D^{-T}\left(\hat u_a+\vec t\right)\cdot\vec\sigma\right].
\end{align}
Combining this expression with
Eq.~\eqref{eq:nonunital-parent} gives
\begin{align}
\tr\left[-\mc N_{D,\vec t}^{-1}(P_a)\Pi_{D,\vec t}(\hat n)\right]=4\alpha_a\left[\hat u_a\cdot\hat n-h_{D,\vec t}(\hat n)\right],
\label{eq:nonunital-translation-cancellation}
\end{align}
where the contribution $\vec t\cdot\hat n$ cancels exactly.

Therefore,
\begin{align}
&\int_{S^2}\max_a\tr\left[-\mc N_{D,\vec t}^{-1}(P_a)\Pi_{D,\vec t}(\hat n)\right]d\mu(\hat n)\notag\\&\ge4\int_{S^2}\max_{a\in\mc A}\alpha_a\left[\hat u_a\cdot\hat n-h_{D,\vec t}(\hat n)\right]d\mu(\hat n)\ge0.
\end{align}
The last inequality follows from
Lemma~\ref{lem:weighted-ellipsoidal-average}, since
$h_{D,\vec t}$ is nonnegative and even, has spherical average $1/2$,
and Eq.~\eqref{eq:nonunital-balance} holds. Thus
Proposition~\ref{prop:orthogonality-normal-form} implies that
$\Pi_{D,\vec t}$ simulates every POVM in the image of
$\mc N_{D,\vec t}^{\dagger}$.

Finally, the result for singular $D$ follows by approximating
$\mc N_{D,\vec t}$ with invertible qubit channels as elaborated in Appendix~\ref{app:singular-general}.
\end{proof}
\subsection{Proof of Corollary~\ref{coro:haar-positive-form}}
\label{app:haar-positive-form}

For the depolarizing channel
\begin{align}
\Delta_r(X)=rX+(1-r)\frac{\tr X}{d}\mbb{1},
\end{align}
the inverse map is
\begin{align}
\Delta_r^{-1}(X)=\frac{1}{r}\left[X-\frac{1-r}{d}\tr(X)\mbb{1}
\right].
\end{align}
The Haar parent is
\begin{align}
\Pi_{\rm Haar}(d\psi)=d\op{\psi}{\psi}d\mu_H(\psi).
\end{align}
Applying Proposition~\ref{prop:orthogonality-normal-form} at
$r=r_d$ gives the condition
\begin{align}
\int\max_a
\tr\left[-\Delta_{r_d}^{-1}(P_a)d\op{\psi}{\psi}\right]
d\mu_H(\psi)\ge0.
\end{align}
Since
\begin{align}
\tr\left[-\Delta_{r_d}^{-1}(P_a)d\op{\psi}{\psi}\right]
=\frac{d}{r_d}\left[\frac{1-r_d}{d}\tr(P_a)-\langle\psi|P_a|\psi\rangle
\right],
\end{align}
and $d/r_d>0$, proposition~\ref{prop:orthogonality-normal-form} is therefore equivalent to
\begin{align}
\int\max_a\left[\frac{1-r_d}{d}\tr(P_a)-\langle\psi|P_a|\psi\rangle\right]d\mu_H(\psi)\ge0
\end{align}
for every finite family $P_a\ge\mbb{0}$ for which there exists a POVM
$\{Q_a\}_a$ satisfying $P_aQ_a=\mbb{0}$ for all $a$. This proves
Corollary~\ref{coro:haar-positive-form}.

\end{document}